\documentclass[prl,twocolumn,aps,floatfix,amsmath,notitlepage]{revtex4-1}

\usepackage{amssymb}
\usepackage{graphicx}
\usepackage{graphics}
\usepackage{amsmath}
\usepackage{amsthm}
\usepackage{color}
\usepackage{dsfont}
\usepackage{mathrsfs}

\usepackage{mathtools}
\usepackage[unicode=true,pdfusetitle,bookmarks=true,bookmarksnumbered=false,bookmarksopen=false,breaklinks=false,pdfborder={0 0 0},pdfborderstyle={},backref=false,colorlinks=true]
{hyperref}

\def\>{\rangle}
\def\<{\langle}

\def\id{\mathsf{id}}

\def\mE{\mathcal{E}}

\def\mF{\mathcal{F}}
\def\mN{\mathcal{N}}
\def\mC{\mathcal{C}}

\def\mV{\mathcal{V}}

\newcommand{\supp}{\operatorname{supp}}

\renewcommand{\qedsymbol}{\nobreak \ifvmode \relax \else
	\ifdim \lastskip<1.5em \hskip-\lastskip \hskip1.5em plus0em
	minus0.5em \fi \nobreak \vrule height0.75em width0.5em
	depth0.25em\fi}

\renewcommand{\geq}{\geqslant}
\renewcommand{\leq}{\leqslant}

\def\mb{\mathfrak{B}}

\def\md{\mathfrak{D}}

\def\mr{\mathfrak{R}}

\def\D{\mathbf{D}}

\newcommand{\locc}{{\rm LOCC}}

\def\pure{{\rm PURE}}

\def\ml{\mathfrak{L}}

\newtheorem{theorem}{Theorem}
\newtheorem*{theorem*}{Theorem}
\newtheorem{corollary}{Corollary}
\newtheorem*{corollary*}{Corollary}
\newtheorem{lemma}{Lemma}
\newtheorem*{lemma*}{Lemma}

\theoremstyle{definition}
\newtheorem{definition}{Definition}
\newtheorem*{definition*}{Definition}

\theoremstyle{remark}
\newtheorem*{remark}{Remark}

\newcommand{\bea}{\begin{eqnarray}}
\newcommand{\eea}{\end{eqnarray}}
\newcommand{\be}{\begin{equation}}
\newcommand{\ee}{\end{equation}}
\newcommand{\ba}{\begin{equation}\begin{aligned}}
\newcommand{\ea}{\end{aligned}\end{equation}}

\newcommand{\bsp}{\begin{split}}
\newcommand{\esp}{\end{split}}

\newtheorem*{example}{Example}

\def\be{\begin{equation}}
\def\ee{\end{equation}}

\newcommand{\rank}{{\rm Rank}}
\newcommand{\cptp}{{\rm CPTP}}
\newcommand{\cp}{{\rm CP}}
\newcommand{\reg}{{\rm reg}}

\newcommand{\lr}{\rangle\langle}
\newcommand{\la}{\langle}
\newcommand{\ra}{\rangle}
\newcommand{\tr}{{\rm Tr}}
\newcommand{\pos}{{\rm Pos}}

\newcommand{\ve}[1]{{\left\vert\kern-0.25ex\left\vert\kern-0.25ex\left\vert #1 
    \right\vert\kern-0.25ex\right\vert\kern-0.25ex\right\vert}}

\newcommand{\mbb}[1]{\mathbb{#1}}

\newcommand{\eqdef}{\coloneqq}

\newcommand{\bb}{\begin{bmatrix}}
\newcommand{\eb}{\end{bmatrix}}

\def\trho{\tilde{\rho}}
\def\tsigma{\tilde{\sigma}}
\def\tomega{\tilde{\omega}}

\def\r{\mathbf{r}}

\def\p{\mathbf{p}}
\def\q{\mathbf{q}}

\def\0{\mathbf{0}}

\def\umd{\underline{\mathfrak{D}}}

\def\bD{\overline{D}}

\def\uM{\underbar{M}}
\def\bM{\overline{M}}
\def\bbd{\overline{\mathbf{D}}}
\def\ubd{\underline{\mathbf{D}}}

\def\mf{\mathfrak{F}}

\usepackage[most]{tcolorbox}
\newtcolorbox{myt}[2][]{%
  attach boxed title to top center
               = {yshift=-4pt},
  colback      = blue!5!white,
  colframe     = blue!75!black,
  halign       = flush left,
  fonttitle    = \bfseries\sffamily,
  colbacktitle = blue!65!black,
  title        = #2,#1,
  enhanced,
}
\newtcolorbox{myd}[2][]{%
  attach boxed title to top center
               = {yshift=-4pt},
  colback      = violet!5!white,
  colframe     = violet!75!black,
  halign       = flush left,
  fonttitle    = \bfseries\sffamily,
  colbacktitle = violet!65!black,
  title        = #2,#1,
  enhanced,
}
\newtcolorbox{mye}[2][]{%
  attach boxed title to top center
               = {yshift=-4pt},
  colback      = purple!5!white,
  colframe     = purple!75!black,
  halign       = flush left,
  fonttitle    = \bfseries\sffamily,
  colbacktitle = purple!65!black,
  title        = #2,#1,
  enhanced,
}

\newtcolorbox{myg}[2][]{%
  attach boxed title to top center
               = {yshift=-4pt},
  colback      = green!5!white,
  colframe     = green!75!black,
  halign       = flush left,
  fonttitle    = \bfseries\sffamily,
  colbacktitle = green!65!black,
  title        = #2,#1,
  enhanced,
}

\newcommand{\eps}{\varepsilon}

\newcommand{\MTadd}[1]{\textcolor{purple}{#1}}

\begin{document}
	
	
	\title{Optimal Extensions of Resource Measures and their Applications}

\author{Gilad Gour}\email{gour@ucalgary.ca}
\affiliation{
Department of Mathematics and Statistics, Institute for Quantum Science and Technology,
University of Calgary, AB, Canada T2N 1N4} 

\author{Marco Tomamichel}
\affiliation{Department of Electrical and Computer Engineering and Centre for Quantum Technologies, National University of Singapore, Singapore}

	\date{\today}
	
	\begin{abstract}
	We develop a framework to extend resource measures from one domain to a larger one. We find that all extensions of resource measures are bounded between two quantities that we call the minimal and maximal extensions. We discuss various applications of our framework. We show that any relative entropy (i.e.\ an additive function on pairs of quantum states that satisfies the data processing inequality) must be bounded by the min and max relative entropies. We prove that the generalized trace distance, the generalized fidelity, and the purified distance are optimal extensions. And in entanglement theory we introduce a new technique to extend pure state entanglement measures to mixed bipartite states.
	\end{abstract}

	\maketitle

\noindent{\it Introduction.} Quite often, significant progress in physics is made by extending certain laws of physics from one domain of applicability to a larger one. A popular example is quantum mechanics itself, which can be viewed as the extension of classical mechanics to the microscopic world (i.e. the quantum domain). General relativity can be viewed in a similar way  as an extension of special relativity to accelerated gravitational systems~\cite{Misner1973}.  Other, more specific examples, include the extension of classical thermodynamics to the quantum domain~\cite{Brandao-2013a}, the extension of pure-state entanglement to mixed state entanglement~\cite{Horodecki-2009a,Plenio-2007a}, the extension of classical Shannon information theory to quantum information~\cite{wilde2013}, and so on.

In each physical theory there are certain quantities that play a major role within the domain of applicability of the theory. Examples include the free energy in thermodynamics, the kinetic energy and work in classical mechanics, entropy, divergences, and channel capacities in information theory, etc. When a physical theory is extended to a larger domain, the relevant quantities that appear in the theory have to be adjusted as well in order to be suitable for the new domain. This raises two compelling questions: (1) ``Under what conditions is the extension of a relevant quantity unique?", and (2) ``Is there a systematic way to construct such extensions that can be applied to many physical theories?"

A suitable framework to study such fundamental questions is the framework of resource theories~\cite{Horodecki-2013a,Coecke-2016,Gour2019}. In this framework, the relevant quantities of the physical theory are described in terms of resource measures (or resource monotones). For example, the free energy in thermodynamics can be viewed as a resource measure since free energy can be used to extract work from a thermodynamical system~\cite{Brandao-2013a}. Similarly, entanglement can be viewed as a resource used for quantum teleportation~\cite{Bennett-1993a}. 
In recent years, resource theories have been developed tremendously, and besides entanglement and quantum thermodynamics, many new resource theories have been identified including the resource theories of asymmetry~\cite{Gour-2008a,Gour-2009a,Marvian-2014a}, Bell non-locality~\cite{Brunner-2014a,deVicente-2014a,Wolfe2019}, coherence~\cite{Streltsov-2016a}, non-Gaussianity~\cite{Ryuji2018,Genoni2010,Marian2013,Albarelli-2018a}, magic~\cite{Howard-2014a,Veitch-2014a}, contextually~\cite{Grudka-2014a}, and many more (see~\cite{Gour2019} for a recent review on the subject). 

In this paper we develop a systematic scheme to extend resource measures from one domain to a larger one~\footnote{After the completion of this work, we became aware that recently, in arXiv:1912.07085, similar extensions were also considered, but with very different applications than the ones we consider in this paper.}. 
To keep the formalism in its most generality, we introduce the concept of \emph{generalized resource theories}.  Our formalism is premised on the fact that resource measures are non-increasing under the set of free operations. We construct two extensions for a given resource measure  and show that they are optimal in the sense that all other extensions of the resource measure in question must be between these two extensions (the minimal and maximal one). We then apply the formalism to quantum divergences, and show that all additive quantum divergences (i.e. relative entropies) must lie between the min and max relative entropies.
Next, we apply the formalism to extend functions from normalized states to sub-normalised states and use it to prove that the generalized fidelity, trace distance, and purified distance are optimal distance measures. Finally, we apply the formalism to entanglement theory and introduce a new way to extend entanglement measures from pure to mixed bipartite states. As an example, we show that the Schmidt number of a mixed bipartite state, as defined in~\cite{Terhal-2000a}, is an optimal (in fact maximal) extension of the Schmidt number of pure bipartite states.


\noindent{\it Notations.} We denote both physical systems and their corresponding Hilbert spaces by $A,B,C$ etc. Classical systems will be denoted by $X$, $Y$, and $Z$.
We only consider finite dimensional systems and denote their dimensions by $|A|$, $|B|$, etc.
The algebra of all $|A|\times|A|$ complex matrices is denoted by $\ml(A)$, and the set of all density matrices in $\ml(A)$ is denoted by $\md(A)$.  For $\rho,\sigma\in\md(A)$ we write $\rho\ll\sigma$ if the support of $\rho$ is a subset of the support of $\sigma$.The set of completely positive maps from $\ml(A)$ to $\ml(B)$ is denoted by $\cp(A\to B)$, and the set of quantum channels by $\cptp(A\to B)$. The one-dimensional trivial physical system is denoted by $1$ (whose associated Hilbert space is $\mbb{C}$). We make the identification $\md(A)=\cptp(1\to A)$.

\noindent{\it General framework.} We start by extending the notion of a quantum resource theory (QRT) to a more general setting in which quantum states and quantum channels are replaced with more abstract objects. Our goal in this generalization is to increase the applicability of our framework.
Let $\mr$ be a mapping that takes any physical system $A$ (e.g.\ atoms, molecules, many body systems, etc) to a set of objects 
$\mr(A)$ (e.g.\ matrices in $\ml(A)$, pair of density matrices, linear maps, etc). Here $\mr(A)$ is replacing the set $\md(A)$ of density matrices which are used in QRT. Note also that $\mr(A)$ can be a subset of $\md(A)$ as will be the case in some of the applications below. 

We also let $\mf$ be a mapping that takes any pairs of physical systems $A$ and $B$, to a set of transformations, $\mf(A\to B)$, from $\mr(A)$ to $\mr(B)$.  In QRTs $\mf(A\to B)$ is the set of free operations which is a subset of quantum channels, but here the maps in $\mf(A\to B)$ do not even need to be linear.
We call the pair $(\mr,\mf)$  a \emph{generalized resource theory} (GRT) if the following two conditions hold:
\begin{enumerate}
\item Doing nothing is free; for any system $A$, $\mf(A\to A)$ contains the identity map.
\item $\mf$ is closed under combination of maps.
\end{enumerate}
Observe that in the definition of a QRT~\cite{Gour2019},  we also identify the set $\mf(A)\eqdef\mf(1\to A)$  as a subset of ``free" objects. Resources are objects in $\mr(A)$ that are not in $\mf(A)$. For our purposes, we will not need to make this identification here. Finally, note that a GRT is a QRT if for any two systems $A$ and $B$, $\mf(A\to B)\subseteq\cptp(A\to B)$ and $\mr(A)=\md(A)$.

In general we do not require $\mf(A\to B)$ to contain only linear transformations. 
As an example of a physical GRT that is non-linear, consider quantum mechanical evolution with post-selection on measurement outcomes.



For a given GRT $(\mr,\mf)$, a function $M:\bigcup_A\mr(A)\to\mbb{R}_+$ is called a resource measure if for all $\rho\in\mr(A)$ and all $\mE\in\mf(A\to B)$
\be\label{monotone}
M\big(\mE(\rho)\big)\leq M(\rho)\;.
\ee
Here we also define a resource measure on smaller subsets. Let $\mr_1$ be a function that maps any physical system $A$ to a subset of objects $\mr_1(A)\subseteq\mr(A)$. Then, we say that 
a function $M_1:\bigcup_A\mr_1(A)\to\mbb{R}$ is an $\mr_1$-resource measure if for all $\rho\in\mr_1(A)$ and all
 $\mE\in\mf(A\to B)$ such that $\mE(\rho)\in\mr_1(B)$, $M_1\big(\mE(\rho)\big)\leq M_1(\rho)$.

The notion of $\mr_1$-resource measure is in fact well known. For example, in entanglement theory, several entanglement measures such as the concurrence were first defined on pure states~\cite{Wootters1998}. Therefore, if $\mr_1(AB)$ is taken to be the set of all pure bipartite states, then in this example $\mr_1$-resource measures are entanglement measures on pure states.

\begin{definition}\label{defopt}
Let $(\mr,\mf)$ be a GRT, $\mr_1$ as above, and $M_1:\bigcup_A\mr_1(A)\to\mbb{R}_+$ be an $\mr_1$-resource measure. The optimal extensions $\bM_1,\uM_1:\bigcup_A\mr(A)\to\mbb{R}_+$, are defined for any $\rho\in\mr(A)$ as follows:
\begin{enumerate}
\item The minimal extension
\be
\uM_1(\rho)\eqdef\sup M_1\big(\mE(\rho)\big)\;,
\ee
where the supremum is over all systems $R$ and all free maps $\mE\in\mf(A\to R)$ that satisfy $\mE(\rho)\in\mr_1(R)$.
If there is no such $\mE$, then $\uM_1(\rho)\eqdef 0$.
\item The maximal extension
\be
\bM_1(\rho)\eqdef\inf M_1(\sigma)\;,
\ee
where the infimum is over all systems $R$ and all 
$\sigma\in\mr_1(R)$ for which there exists $\mE\in\mf(R\to A)$ that satisfies
$\rho=\mE(\sigma)$.
If there is no such $\sigma$, then $\bM_1(\rho^A)\eqdef+\infty$.
\end{enumerate}
\end{definition}

Roughly speaking, $\overline{M}_1$ can be interpreted as the $\mr_1$-resource cost of $\rho$ as measured by $M_1$, and $\underline{M}_1$ can be interpreted as the distillable $\mr_1$-resource as measured by $M_1$ (see \MTadd{sketch} in~FIG.~\ref{definition}).

\begin{figure}[h]\centering
    \includegraphics[width=0.47\textwidth]{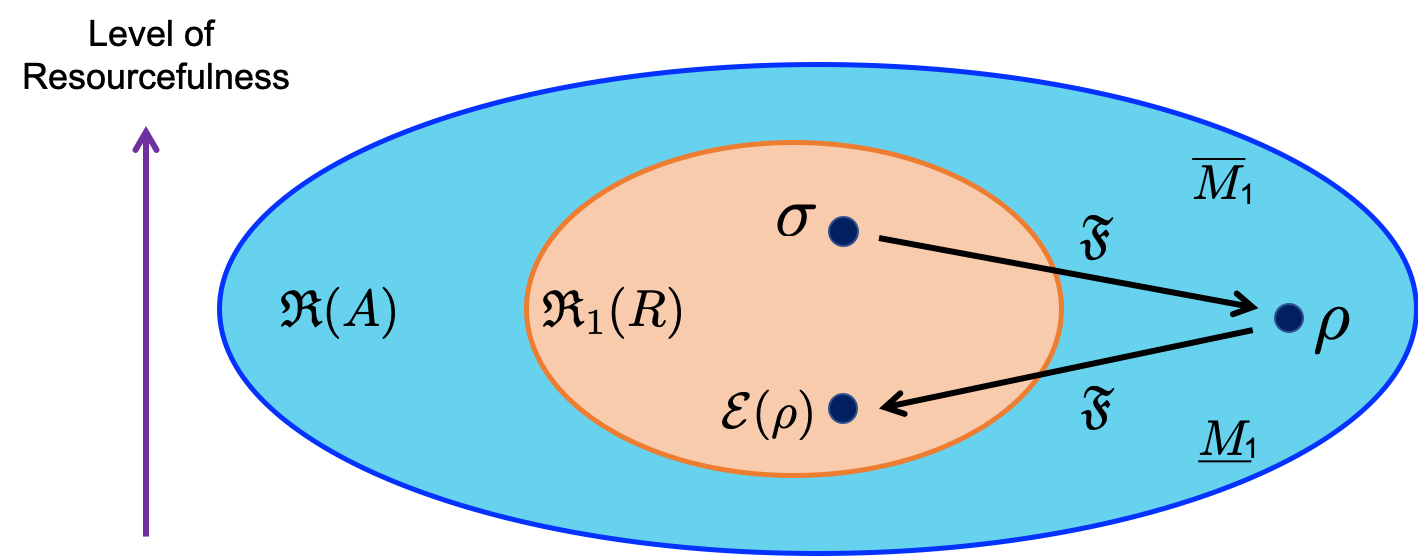}
  \caption{\linespread{1}\selectfont{\small Minimal and maximal extensions of $M_1$. The function $\overline{M}_1(\rho)$ is defined as the minimization of $M_1(\sigma)$ over all $\sigma\in\mr_1(R)$ that can be converted by free operations to $\rho$.
  The function $\uM_1(\rho)$ is defined as the maximization of $M_1(\sigma)$ over all $\sigma=\mE(\rho)$ that can be reached from $\rho$ by free operations.}}
  \label{definition}
\end{figure}

\begin{theorem}\label{properties}
Let $(\mr,\mf)$, $\mr_1$ and $M_1$ as in Definition~\ref{defopt}. The  extensions $\bM_1$ and $\uM_1$ have the following properties:
\begin{enumerate} 
\item \textbf{Reduction.} 
\be
\uM_1(\rho)=\bM_1(\rho)=M_1(\rho)\quad\forall\;\rho\in\mr_1(A).
\ee
\item \textbf{Monotonicity.} For any $\rho\in\mr(A)$ and $\mE\in\mf(A\to B)$
\be
\uM_1\big(\mE(\rho)\big)\leq\uM_1(\rho)\quad\text{and}\quad\bM_1\big(\mE(\rho)\big)\leq\bM_1(\rho)\;.
\ee
\item \textbf{Optimality.} Any resource measure $M:\bigcup_A\mr(A)\to\mbb{R}$ with $M(\sigma)=M_1(\sigma)$ for all 
$\sigma\in\mr_1(A)$, must satisfy
\be\label{bounds9}
\uM_1(\rho)\leq M(\rho)\leq \bM_1(\rho)\quad\forall\rho\in\mr(A)\;.
\ee
\end{enumerate}
\end{theorem}

The optimality property above implies that equality between $\uM_1$ and $\bM_1$ would imply uniqueness of all extensions of $M_1$. While one does not expect it to happen in general, it does occur in some GRTs as we will see below. We postpone the details of all proofs to the supplemental material (SM).

We say that a GRT admits a tensor product structure if both $\mr$ and $\mf$ are closed under tensor products. For such a GRT, the theorem above can also be applied to resource measures that are additive under tensor products. In general, the additivity property of a resource measure does not carry over to its minimal and maximal extensions. Instead, in the SM we show that for any additive measure $M_1$, $\rho\in\mr(A)$, and $\sigma\in\mr(B)$
\ba
&\bM_1(\rho\otimes\sigma)\leq\bM_1(\rho)+\bM_1(\sigma)\\
&\uM_1(\rho\otimes\sigma)\geq\uM_1(\rho)+\uM_1(\sigma)\;.
\ea
That is, $\bM_1$ is sub-additive and $\uM_1$ is super-additive. This sub/super additivity of the optimal extensions holds even if 
$M_1$ is only weakly (or partially) additive; i.e. for all $\rho\in\mr_1(A)$
\be
M_1(\rho^{\otimes k})=kM_1(\rho)\quad\quad\forall\;k\in\mbb{N}.
\ee
This means that for the regularized extensions, given by 
\ba\label{regm}
&\bM_1^{\reg}(\rho)\eqdef\lim_{n\to\infty}\frac{1}{n}\bM_1(\rho^{\otimes n})\\
&\uM_1^{\reg}(\rho)\eqdef\lim_{n\to\infty}\frac{1}{n}\uM_1(\rho^{\otimes n})\;,
\ea 
the limits above exists (see SM for more details).

\begin{theorem}\label{thm2}
Let $(\mr,\mf)$ be a GRT that admits a tensor product structure, $\mr_1$ closed under tensor products, and $M_1:\bigcup_A\mr_1(A)\to\mbb{R}_+$ be a weakly additive $\mr_1$-resource measure. Then, the optimal regularized extensions $\bM_1^\reg$ and $\uM_1^\reg$ satisfy the reduction, monotonicity, and optimality properties of Theorem~\ref{properties} with $\bM_1$ and $\uM_1$ replaced by $\bM_1^\reg$ and $\uM_1^\reg$, respectively. Moreover,
$\bM_1^\reg$ and $\uM_1^\reg$ are weakly additive.
\end{theorem}

Since $\uM_1(\rho)\leq\uM_1^{\reg}(\rho)$ and $\bM_1(\rho)\geq \bM_1^{\reg}(\rho)$ (see SM), the optimality property on a weakly additive extension $M$ is in general stronger (i.e.\ tighter) than the bounds given in~\eqref{bounds9} for non-additive measures.

\noindent{\it Applications to quantum divergences.}
Consider the GRT $(\mr,\mf)$, where $\mr(A)\eqdef\{(\rho,\sigma)\;:\;\rho,\sigma\in\md(A)\}$ consists of pairs of density matrices in $\md(A)$, and the set $\mf(A\to B)$ consists of all pairs of quantum channels of the form $(\mE,\mE)$, where $\mE\in\cptp(A\to B)$. That is, under free operations the pair $(\rho,\sigma)$, where $\rho,\sigma\in\md(A)$, transforms to the pair $\big(\mE(\rho),\mE(\sigma)\big)$. The resource measures in this GRT are called quantum divergences, and classical divergences can be viewed as $\mr_1$-resource measures, where $\mr_1(X)\eqdef\{(\p,\q)\;:\;\p,\q\in\md(X)\}$ consists of pair of probability vectors (which we view as diagonal density matrices in a fixed classical basis). We can therefore apply the techniques introduced above to extend classical divergences to quantum divergences. We start with a formal definition of a quantum divergence and a relative entropy (cf.~\cite{GT2020a,Matsumoto2018b,WGE2017,Muller2018,BEG2019,RW2019,Hall1999,Hall2000}).

\begin{definition}\label{qdefd}
Let $\D:\bigcup_{A}\big\{\md(A)\times\md(A)\big\}\to\mbb{R}_{+}\cup\{\infty\}$ be a function acting on pairs of $|A|$-dimensional probability vectors in all finite dimensions.
\begin{enumerate}
\item The function $\D$ is called a \emph{quantum divergence} if it satisfies the data processing inequality (DPI); i.e.  for any $\rho,\sigma\in\md(A)$ and a quantum channel $\mE\in\cptp(A\to B)$
\be
\D\big(\mE(\rho)\big\|\mE(\sigma)\big)\leq \D(\rho\|\sigma)\;.
\ee
\item A quantum divergence $\D$ is called a \emph{quantum relative entropy} if in addition it satisfies:
\begin{enumerate}
\item Additivity. For any $\rho_1,\rho_2\in\md(A)$ and any $\sigma_1, \sigma_2\in\md(B)$
\be
\D\left(\rho_1\otimes\rho_1\big\|\sigma_2\otimes\sigma_2\right)= \D(\rho_1\|\sigma_1)+\D(\rho_2\|\sigma_2)\;.
\ee
\item Normalization. 
\be
\D\Big(|0\lr 0|\Big\|\frac12|0\lr 0|+\frac12|1\lr 1|\Big)=1\;.
\ee
\end{enumerate}
\end{enumerate}
\end{definition}

Two extreme cases of relative entropies that play an important role particularly in single-shot quantum information are the min and max relative entropies~\cite{Datta-2009a} defined for any pair of states $\rho,\sigma\in\md(A)$ with $\supp(\rho)\subseteq\supp(\sigma)$ as
\begin{align}
&D_{\min}(\rho\|\sigma)\eqdef-\log\tr\left[\Pi_\rho\sigma\right]\\
&D_{\max}(\rho\|\sigma)\eqdef\log\min\{t\in\mbb{R}\;:\;t\sigma\geq\rho\}\;,
\end{align}
and for $\supp(\rho)\not\subseteq\supp(\sigma)$, $D_{\max}(\rho\|\sigma)=+\infty$, whereas $D_{\min}$ remains the same unless $\rho$ and $\sigma$ are orthogonal in which case $D_{\min}(\rho\|\sigma)\eqdef\infty$.

\begin{theorem}\label{thmmm}
Let $\D$ be a quantum relative entropy as in Definition~\ref{qdefd}.
Then, 
for all $\rho,\sigma\in\md(A)$, 
\be\label{32}
D_{\min}(\rho\|\sigma)\leq \D(\rho\|\sigma)\leq D_{\max}(\rho\|\sigma)\;.
\ee
\end{theorem}

Our definition of a quantum relative entropy has in it enough structure, that it gives rise to several additional properties. For example, in the SM we use the bounds above to provide an alternative proof to the one given in~\cite{Matsumoto2018b}, singling out the Umegaki relative entropy as the unique relative entropy that is asymptotically continuous.
In the SM we also show that any quantum relative entropy $\D$ satisfies $\forall\rho,\sigma,\omega\in\md(A)$,
\be
\D(\rho\|\sigma)\leq\D(\rho\|\omega)+D_{\max}(\omega\|\sigma)\;.
\ee
This triangle inequality implies for example that $\D$ is continuous in its second argument on states with full rank.
In addition we also show that quantum divergences are continuous in their first argument, and that most of them are faithful (i.e. $D(\rho\|\sigma)=0$ implies $\rho=\sigma$).
We now introduce the optimal extensions of a classical divergence to the quantum domain.

Let $\D_1:\md(X)\times\md(X)\to \mbb{R}_+$ be a classical divergence. According to Definition~\ref{defopt}, the minimal and maximal extensions of $\D_1$ to the quantum domain $\md(A)\times\md(A)$ are defined, respectively, for any $\rho,\sigma\in\md(A)$ by
\begin{align}
&\ubd_1(\rho\|\sigma) \eqdef \sup \D_1\big(\mE\left(\rho\right)\|\mE(\sigma)\big) , \\
&\bbd_1(\rho\|\sigma)\eqdef\inf \big\{\D_1(\p\|\q):\rho=\mF(\p),\;\sigma=\mF(\q) \big\} ,
\end{align}
where the optimizations are over the classical system $X$, the channels $\mE \in \cptp(A \to X)$ and $\mF \in \cptp(X \to A)$ as well as the probability distributions $\p$ and $\q$.

From Theorem~\ref{properties} it follows that both $\ubd_1$ and $\bbd_1$ are quantum divergences that reduces to $\D_1$ on classical pair of states in $\md(X)\times\md(X)$. Moreover, the main significance of these divergences follows from the third part of Theorem~\ref{properties}. It implies that  any quantum divergence
$\D:\md(A)\times\md(A)\to \mbb{R}_+$ that reduces on classical states to a classical divergence $\D_1$ satisfies for all $\rho,\sigma\in\md(A)$
\be
\ubd_1(\rho\|\sigma)\leq \D(\rho\|\sigma)\leq\bbd_1(\rho\|\sigma)\;,
\ee
where $\ubd_1,\bbd_1$ are the minimal and maximal quantum extensions of $\D_1$.

In general, $\ubd_1$ and $\bbd_1$ are not additive even if $\D_1$ is a relative entropy. However, one can regularize these quantities to obtain a divergence that is at least weakly additive. For example, in~\cite{MO2014,HT2016,Tomamichel2015} it was shown that if $\D_1=D_\alpha$ is the classical R\'enyi divergence with $\alpha\geq\frac12$, then the regularized minimal extension of $D_{\alpha}$, given by
\be
\underline{D}_{\alpha}^{\reg}(\rho\|\sigma)\eqdef\lim_{n\to\infty}\frac1n\underline{D}_{\alpha}\left(\rho^{\otimes n}\big\|\sigma^{\otimes n}\right)\;,
\ee 
has a closed formula given by the sandwiched or minimal quantum R\'enyi divergence~\cite{Muller-Lennert-2013a,Wilde2014}. For $\alpha\in[0,1/2)$ one can use the symmetry 
$D_{\alpha}(\p\|\q)=\frac{\alpha}{1-\alpha}D_{\alpha}(\q\|\p)$ to to find a close form~\footnote{We owe this point to Mil\'an Mosonyi.}. Thus, for all $\alpha\in[0,\infty]$ the minimal extension of the R\'enyi divergence, or minimal quantum R\'enyi divergence, is given by (cf.~\cite{Matsumoto2016})
	\begin{widetext}
	\be
	\underline{D}_{\alpha}(\rho\|\sigma)=\\
	\begin{cases}
	\frac1{\alpha-1}\log\tr\left(\sigma^{\frac{1-\alpha}{2\alpha}}\rho\sigma^{\frac{1-\alpha}{2\alpha}}\right)^\alpha&\text{if }\big(\frac12\leq\alpha<1\wedge\rho\not\perp\sigma\big)\vee\rho\ll\sigma\\
	\frac1{\alpha-1}\log\tr\left(\rho^{\frac{\alpha}{2(1-\alpha)}}\sigma\rho^{\frac{\alpha}{2(1-\alpha)}}\right)^{1-\alpha}&\text{if }0\leq \alpha<\frac12\wedge\rho\not\perp\sigma\\
	\infty&\text{otherwise}
	\end{cases}
	\ee
	\end{widetext}
where $\wedge$ and $\vee$ stands for `and' and `or', respectively.
Note that the above formula is continuous in $\alpha$. 

For the maximal extension we have the following result.

\begin{theorem}\label{thmmax}
Let $\D_1$ be a classical relative entropy, and let $\bbd_1$ be its maximal extension to quantum states. Then, for any pure state $\psi\eqdef |\psi\lr\psi|$ and any mixed state $\sigma$,
\be
{\bf \bD}_1(\psi\big\|\sigma)=D_{\max}(\psi\big\|\sigma)=\log\la\psi|\sigma^{-1}|\psi\ra\;.
\ee 
\end{theorem}
\begin{remark}
The maximal divergence is equivalent to the expression considered in~\cite{Matsumoto2018} (and more recently in~\cite{Tomamichel2015,Katariya2020}). There, a closed formula was found for the maximal extension when $\D_1=D_\alpha$ and $\alpha\in[0,2]$. In this case the maximal extension becomes the geometric relative entropy given by (assuming $\supp(\rho)\subseteq\supp(\sigma)$)
$$
\overline{D}_\alpha(\rho\|\sigma)=\widehat{D}_\alpha(\rho\|\sigma)\eqdef\frac{1}{\alpha-1}\tr\left[\sigma\left(\sigma^{-\frac12}\rho\sigma^{-\frac12}\right)^\alpha\right]\;.
$$
In Proposition 66 of~\cite{Katariya2020} it was shown that this result is consistent with Theorem~\ref{thmmax}.  Note, however, that Theorem~\ref{thmmax} holds for all choices of $\D_1$ even for $\D_1=D_{\alpha}$ with $\alpha>2$. However, for $\alpha>2$ it is left open to find a closed formula for $\bD_\alpha(\rho\|\sigma)$ (and its regularization in case it is not additive) when $\rho$ is a mixed state.
\end{remark}

\noindent{\it Extensions to sub-normalized states.}
One of the properties of quantum channels is that they take normalized states to normalized states. When considering sub-normalized states,
all trace non-increasing (TNI) CP maps (including CPTP maps) take sub-normalized states to subnormalized states.
In applications, it is quite often useful to quantify distances between subnormalized states with a function that obeys a monotonicity  property (i.e. the data processing inequality) under TNI-CP maps. We therefore consider here the GRT $(\mr,\mf)$, where $\mr(A )$ consists of all pairs of sub-normalized states, and $\mf(A\to B )$ is the set of all  pairs $(\mE,\mE)$, where $\mE$ is a TNI-CP map. We also take $\mr_1(A)\subseteq\mr(A)$ to be the subset of all pairs of normalized states. With these choices, the maximal extension of a quantum divergence $\D$ to subnormalized states is given for all subnormalized states $(\trho,\tsigma)\in\mr(A)$ by
\be\label{dextd}
\bbd(\trho,\tsigma)\eqdef\inf \D(\rho,\sigma)
\ee
where the infimum is over all systems $R$ and all density matrices $\rho,\sigma\in\md(R)$ for which there exists
a TNI-CP map $\mE\in\cp(R\to A)$ such that $\trho=\mE(\rho)$ and $\tsigma=\mE(\sigma)$. Note that this divergence satisfies the DPI with TNI-CP maps (not necessarily CPTP maps), and that we do not consider the minimal extension in this case, since it is zero. Remarkably, the maximal extension has the following simple closed formula.

\begin{theorem}\label{gmain}
Let $\D$ be a quantum divergence and $\bbd$ be its maximal extension to sub-normalized states. For any pair of sub-normalized states $(\trho,\tsigma)\in\mr(A )$
\be\label{g1}
\bbd(\trho\|\tsigma)=\D\Big(\trho\oplus\big(1-\tr[\trho]\big)\;\big\|\;\tsigma\oplus\big(1-\tr[\tsigma]\big)\Big)\;.
\ee
\end{theorem}

Note that if $\D$ is the trace-distance or the fidelity then $\bbd$ coincides with the generalized trace distance or the generalized fidelity~\cite{tomamichel09}, respectively. This means that the generalized trace distance as defined in~\cite{Tomamichel2015}, for example, is optimal in the sense that any distance measure of subnormalized states that satisfies DPI under TNI-CP maps and that reduces to the trace distance on density matrices must be smaller than the generalized trace distance. Moreover, in the SM we show that the purified distance as defined in~\cite{tomamichel09}, is the maximal extension of the trace distance from pure states to subnormalized mixed states. This gives an operational meaning to the purified distance as the largest metric-divergence that reduces to the trace distance on pure states.

\noindent{\it Applications in entanglement theory.}
We apply now the techniques above to extend entanglement measures from pure bipartite states states to mixed bipartite states.  In general, our extensions are different than the convex-roof extensions (that have been used extensively in literature -see e.g.~\cite{Horodecki-2009a,Plenio-2007a}) and consequently, they enlarge the toolbox of entanglement theory. 
We use the notation $\pure(AB)\subset\md(AB)$ to denote the set of all pure states. On pure states we will call $E$ entanglement measure if it does not increase under  Local Operations and Classical Communication (LOCC).

For any function on pure states $E_1:\bigcup_{A,B}\pure(AB)\to\mbb{R}$  we define its maximal  extension to mixed bipartite states as
\be\label{me}
\overline{E}_1(\rho^{AB})\eqdef\inf E_1(\psi^{A'B'})
\ee
where the infimum is over all systems $A',B'$, and all pure states $\psi\in\pure(A'B')$ for which there exists $\mE\in\locc(A'B'\to AB)$ such that $\rho^{AB}=\mE(\psi^{A'B'})$.
The minimal extension is defined as
\be
\underline{E}_1(\rho^{AB})\eqdef\sup E_1\left(\mE(\rho^{AB})\right)
\ee
where the supremum is over all systems $A',B'$, and all $\mE\in\locc(AB\to A'B')$ such that $\mE(\rho)\in\pure(A'B')$.

The maximal and minimal extensions of a pure-state entanglement measure $E_1$, can be interpreted as the zero-error pure-entanglement cost and distillation of $\rho^{AB}$, respectively. Note also that the minimal extension is quite often zero since it is not alway possible to find an LOCC channel that can be used to obtain a pure entangled state from $\rho^{AB}$. To avoid that, in the SM we also introduce the smoothed version of these quantities.

A direct corollary of Theorem~\ref{properties} implies that the minimal and maximal extensions of  a measure of entanglement on pure states, $E_1$, satisfy: 
\begin{enumerate}
\item For any $\psi\in{\rm \pure}(AB)$ 
\be
\overline{E}_1(\psi^{AB})=\underline{E}_1(\psi^{AB})=E_1(\psi^{AB})\;.
\ee
\item For any $\rho\in\md(AB)$ and $\mE\in\locc(AB\to A'B')$
\ba
&\overline{E}_1\left(\mE^{AB\to A'B'}(\rho^{AB})\right)\leq \overline{E}_1\left(\rho^{AB}\right)\\
&\underline{E}_1\left(\mE^{AB\to A'B'}(\rho^{AB})\right)\leq \underline{E}_1\left(\rho^{AB}\right)\;.
\ea
\item For any measure of entanglement $E$ that reduces to $E_1$ on pure states
\be
\underline{E}_1\left(\rho^{AB}\right)\leq E\left(\rho^{AB}\right)\leq \overline{E}_1\left(\rho^{AB}\right)\;.
\ee
\end{enumerate}

\begin{example}
The Schmidt number on a pure state $\psi\in\pure(AB)$ is defined as
\be
N(\psi^{AB})\eqdef \rank(\psi^A)
\ee
where $\psi^A$ is the reduced density matrix of $\psi^{AB}$. For a density matrix $\rho\in\md(AB)$, the Schmidt number was define in~\cite{Terhal-2000a} as the number $k$ that satisfies: (1) for any decomposition of $\rho^{AB}=\sum_ip_i|\psi_i\lr\psi_i|$, $p_i\geq 0$
at least one of the vectors $\psi_i$ has at least
Schmidt rank $k$, and (2) there exists a decomposition of $\rho^{AB}$
with all vectors $\psi_i^{AB}$ of Schmidt rank at most $k$. In the SM we show that this definition coinside
with the maximal extension~\eqref{me} when $E_1$ is taken to be the Schmidt number on pure states. Therefore, the third part of Theorem~\ref{properties} implies that the Schmidt number as defined in~\cite{Terhal-2000a} is optimal in the sense that any other measure of entanglement that reduces to the Schmidt number on pure states must be no greater than it.
\end{example}

\noindent{\it Conclusions.}
In this paper we developed a simple (yet powerful) framework to extend resource monotones from one domain to a larger one.
We then applied this framework to several theories in physics, focusing on quantum divergences that plays a key role in quantum information and QRTs. We were able to use it to show a fundamental property about quantum relative entropies, they all lies between the min and max relative entropies. Our framework also demonstrated the existence of quantum divergences that are at least weakly additive and that to the authors' knowledge were not discussed in literature before.
We then applied the formalism to extend distance measures from normalized to subnormalized states, showing the optimality of the generalized fidelity, trace distance, and purified distance, and finally developed a new method to extend entanglement measures from pure to mixed bipartite states. 

The extension framework presented in this paper answer both a fundamental question regarding the uniqueness of the extensions of resource monotones (see third part of Theorem~\ref{properties} and Theorem~\ref{thm2}), as well as the practical question regarding the construction of optimal extensions (see Definition~\ref{defopt} and Eq.~\eqref{regm}).
As these questions lies at the heart of many physical theories, the range of the applications that were explored in this paper only touches the tip of the iceberg. We expect our extension framework to have numerous applications, and in~\cite{Gour2020} its applications to channel divergences and dynamical resource theories will be explored. Other applications that were not explored here due to the space limit, includes the QRT of coherence, asymmetry, athermality, and many other QRTs. We leave these investigations for future work.

\begin{acknowledgments}
GG acknowledge support from the Natural Sciences and Engineering Research Council of Canada (NSERC).
MT is supported by NUS startup grants (R-263-000-E32-133 and R-263-000-E32-731) and by the National Research Foundation, Prime Minister's Office, Singapore and the Ministry of Education, Singapore under the Research Centres of Excellence programme.
\end{acknowledgments}

\bibliographystyle{apsrev4-1}
\bibliography{QRTbib}

\newpage
\onecolumngrid


\begin{center}\Large\bfseries
 Supplemental Material\\ Optimal Extensions of Resource Measures and their Applications 
\end{center}

\section{Proof of Theorem~\ref{properties}}

\begin{theorem*}[Theorem~\ref{properties} in the main text]
Let $(\mr,\mf)$ be a GRT and for any system $A$, let $\mr_1(A)\subseteq\mr(A)$, and consider an $\mr_1$-resource measure $M_1:\bigcup_A\mr_1(A)\to\mbb{R}$. The optimal extensions $\bM_1,\uM_1:\bigcup_A\mr(A)\to\mbb{R}$ have the following three properties:
\begin{enumerate} 
\item \textbf{Reduction.} For any $\rho\in\mr(A)$
\be
\uM_1(\rho^A)=\bM_1(\rho^A)=M_1(\rho^A)\quad\forall\;\rho\in\mr_1(A).
\ee
\item \textbf{Monotonicity.} For any $\rho\in\mr(A)$ and $\mE\in\mf(A\to B)$
\be
\uM_1\big(\mE(\rho)\big)\leq\uM_1(\rho)\quad\text{and}\quad\bM_1\big(\mE(\rho)\big)\leq\bM_1(\rho)\;.
\ee
\item \textbf{Optimality.} Any resource measure $M:\bigcup_A\mr(A)\to\mbb{R}$ with $M(\sigma)=M_1(\sigma)$ for all 
$\sigma\in\mr_1(A)$, must satisfy
\be\label{bounds9}
\uM_1(\rho)\leq M(\rho)\leq \bM_1(\rho)\quad\forall\rho\in\mr(A)\;.
\ee
\end{enumerate}
\end{theorem*}

\begin{proof}
\textbf{Reduction.} If $\rho\in\mr_1(A)$ then $\uM_1(\rho)\geq M_1(\rho)$ since $\id^A\in\mf(A\to A)$.
Conversely, since $M_1$ is an $\mr_1$-resource measure and $\rho\in\mr_1(A)$, for any $\mE\in\mf(A\to R)$ with $\mE(\rho)\in\mr_1(R)$,  we have $M_1\big(\mE(\rho)\big)\leq M_1(\rho)$. Hence,  $\uM_1(\rho)\leq M_1(\rho)$.

For $\bM_1$, observe again that by taking $R=A$ and $\mE=\id^A$ we get $\bM_1(\rho)\leq M_1(\rho)$.
Conversely, since $\rho\in\mr_1(A)$ and $M_1$ is an $\mr_1$-resource measure, for any $\sigma\in\mr_1(R)$ and $\mE\in\mf(R\to A)$ such that $\rho=\mE(\sigma)$ we have $M_1(\rho)\leq M_1(\sigma)$. Hence, $\bM_1(\rho)\geq M_1(\rho)$.

\textbf{Monotonicity.} For any $\rho\in\mr(A)$ and 
$\mN\in\mf(A\to B)$ we have
\ba
\uM_1\big(\mN(\rho)\big)&=\sup_R\Big\{M_1\big(\mE\circ\mN(\rho)\big)\;:\; \mE\in\mf(B\to R)\;,\;\mE\circ\mN(\rho)\in\mr_1(R)\Big\}\\
&\leq \sup_R\Big\{M_1\big(\mF(\rho)\big)\;:\; \mF\in\mf(A\to R)\;,\;\mF(\rho)\in\mr_1(R)\Big\}\\
&=\uM_1(\rho)\;,
\ea
where the inequality follows by replacing $\mE\circ\mN\in\mf(A\to R)$ with any map $\mF\in\mf(A\to R)$ with the same property that $\mF(\rho)\in\mr_1(R)$.

We also have
\ba
\bM_1(\rho^A)&=\inf_{R}\Big\{M_1(\sigma)\;:\; \sigma\in\mr_1(R)\;,\;
\rho=\mE(\sigma)\;,\; \mE\in\mf(R\to A)\Big\} \\
&\geq\inf_{R}\Big\{M_1(\sigma)\;:\; \sigma\in\mr_1(R)\;,\;
\mN(\rho)=\mN\circ\mE(\sigma)\;,\; \mE\in\mf(R\to A)\Big\}\\
&\geq \inf_{R}\Big\{M_1(\sigma)\;:\; \sigma\in\mr_1(R)\;,\;
\mN(\rho)=\mF(\sigma)\;,\; \mF\in\mf(R\to B)\Big\}\\
&=\bM_1\big(\mN(\rho)\big)
\ea

\textbf{Optimality.} Let $\rho\in\mr(A)$ and $\mE\in\mf(A\to R)$ such that $\mE(\rho)\in\mr_1(R)$. Then, from the monotonicity of $M$ we get 
\be
M(\rho)\geq M\big(\mE(\rho)\big)=M_1\big(\mE(\rho)\big)\;.
\ee
Since the above holds for all systems $R$ and all $\mE\in\mf(A\to R)$ with $\mE(\rho)\in\mr_1(R)$ we must have 
$\uM_1(\rho)\leq M(\rho)$.

For the second inequality, let $\rho\in\mf(A)$, $\sigma\in\mf_1(R)$, and suppose there exists  $\mE\in\mf(R\to A)$ such that $\rho=\mE(\sigma)$. Then, from the monotonicity of $N$ we get
\be
M(\rho)= M\big(\mE(\sigma)\big)\leq M(\sigma)=M_1(\sigma)\;.
\ee
Since the above inequality holds for all such $\sigma\in\mf_1(R)$ for which there exists  $\mE\in\mf(R\to A)$
that takes $\sigma$ to $\rho$, we conclude that $M(\rho)\leq\bM_1(\rho)$.
\end{proof}

\section{Sub-additive and super-additive properties of optimal extensions}

\begin{lemma*}
Let $(\mr,\mf)$ be a GRT that admits a tensor product structure, and let $\mr_1(A)\subseteq\mr(A)$ for any system $A$, and suppose that it is closed under tensor products.
Also, let $M_1:\bigcup_A\mr_1(A)\to\mbb{R}$ be  an $\mr_1$-resource measure that is additive under tensor products.  
Then, its minimal and maximal extensions to $\mr$, satisfy for any $\rho\in\mr(A)$ and $\sigma\in\mr(B)$
\be
\bM_1(\rho\otimes\sigma)\leq\bM_1(\rho)+\bM_1(\sigma)\quad\text{and}\quad\uM_1(\rho\otimes\sigma)\geq\uM_1(\rho)+\uM_1(\sigma)\;.
\ee
That is, $\bM_1$ is sub-additive and $\uM_1$ is super-additive.
\end{lemma*}

\begin{proof}
By definition,
\ba
&\uM_1(\rho\otimes\sigma)\eqdef\sup_R\Big\{M_1\big(\mE(\rho\otimes\sigma)\big)\;:\; \mE\in\mf(AB\to R)\;,\;\mE(\rho\otimes\sigma)\in\mr_1(R)\Big\}\\
&\geq\sup_R\Big\{M_1\big(\mE_1(\rho)\otimes\mE_2(\sigma)\big)\;:\; \mE_1\in\mf(A\to R_1)\;,\;\mE_2\in\mf(B\to R_2)\;,\;\mE_1(\rho)\in\mr_1(R_1)\;,\;\mE_2(\sigma)\in\mr_1(R_2)\Big\}\\
&=\uM_1(\rho)+\uM_1(\sigma)\;,
\ea
where the inequality follows from retricting $\mE$ to the form $\mE_1\otimes\mE_2$, and $R$ to the composite form $R_1\otimes R_2$. The sub-additivity of $\bM_1$ follows similar lines.
\end{proof}

Suppose now that $M_1$ is weakly additive; i.e. for all $\rho\in\mr_1(A)$
\be
M_1(\rho^{\otimes k})=kM_1(\rho)\quad\quad\forall\;k\in\mbb{N}.
\ee
Then, similar arguments that were used in the proof above can be used to show that for any
$\rho\in\mr_1(A)$
\be
\bM_1(\rho^{\otimes (k+\ell)})\leq \bM_1(\rho^{\otimes k})+\bM_1(\rho^{\otimes \ell})\quad\text{and}\quad\uM_1(\rho^{\otimes (k+\ell)})\geq \uM_1(\rho^{\otimes k})+\uM_1(\rho^{\otimes l})\quad\quad\forall\;k,\ell\in\mbb{N}.
\ee 
The above inequalities grantee that the limit in~\eqref{regm} of the regularized extensions exists. The also implies that
\be
\bM_1^\reg(\rho)\leq \bM_1(\rho)\quad\text{and}\quad\uM_1^\reg(\rho)\geq \uM_1(\rho)\quad\quad\forall\;\rho\in\mr(A)\;.
\ee

\section{Proof of Theorem~\ref{thm2}}

\begin{theorem*}[detailed version of Theorem~\ref{thm2}]
Let $(\mr,\mf)$ be a GRT that admits a tensor product structure, $\mr_1$ as above, and $M_1:\bigcup_A\mr_1(A)\to\mbb{R}$ be a weakly additive $\mr_1$-resource measure (we assume that $\mr_1$ is closed under tensor products). Then, the optimal regularized extensions $\bM_1^\reg,\uM_1^\reg:\bigcup_A\mr(A)\to\mbb{R}$ satisfy the following properties:
\begin{enumerate}
\item Reduction. For all $\rho\in\mr_1(A)$
\be\uM_1^{\reg}(\rho)=\bM_1^{\reg}(\rho)=M_1(\rho)\;.\ee
\item Monotonicity. For any $\rho\in\mr(A)$ and $\mE\in\mf(A\to B)$ 
\ba
&\uM_1^{\reg}\big(\mE(\rho)\big)\leq \uM_1^{\reg}(\rho)\\
&\bM_1^{\reg}\big(\mE(\rho)\big)\leq\bM_1^{\reg}(\rho)\;.
\ea
\item Weak Additivity. For all $\rho\in\mr_1(A)$ and $k\in\mbb{N}$
\ba
&\uM_1^{\reg}(\rho^{\otimes k})=k\uM_1^{\reg}(\rho)\\
&\bM_1^{\reg}(\rho^{\otimes k})=k\bM_1^{\reg}(\rho)\;.
\ea
\item Optimality. For any weakly additive extension $M$ of $M_1$, from $\mr_1$ to $\mr$, that is monotonic under elements of $\mf$, we have
\be
\uM_1^{\reg}(\rho)\leq M(\rho)\leq \bM_1^{\reg}(\rho)\quad\forall\rho\in\mr(A)\;.
\ee
\end{enumerate}
\end{theorem*}

Note that $\uM_1(\rho)\leq\uM_1^{\reg}(\rho)$ and $\bM_1(\rho)\geq \bM_1^{\reg}(\rho)$ so that the bounds above on a weakly additive extension $N$ are in general tighter than the bounds given in~\eqref{bounds9} on non-additive measures.

\begin{proof}
Since $\rho\in\mr_1(A)$ and $\mr_1$ is closed under tensor products it follows that $\rho^{\otimes n}\in\mr_1(A^n)$. 
Since $\uM_1$ and $\bM_1$ are extensions of $M_1$ to $\mr$, they reduce to $M_1$ on elements in $\mr_1(A)$. Hence, for any $n\in\mbb{N}$ we have  $\uM_1(\rho^{\otimes n})=\bM_1(\rho^{\otimes n})=M_1(\rho^{\otimes n})=nM_1(\rho)$. This implies that 
$\uM_1^{\reg}(\rho)=\bM_1^{\reg}(\rho)=M_1(\rho)$.

For the monotonicity property, observe that for any $n\in\mbb{N}$, $\rho^{\otimes n}\in\mr(A^n)$ and $\mE^{\otimes n}\in\mf(A^n\to B^n)$ so that
\be
\frac{1}{n}\uM_1\left(\big(\mE(\rho)\big)^{\otimes n}\right)=\frac{1}{n}\uM_1\left(\mE^{\otimes n}\big(\rho^{\otimes n}\big)\right)\leq\frac{1}{n}\uM_1\big(\rho^{\otimes n}\big)
\ee
where the inequality follows from the monotonicity of $\uM_1$. Hence, taking the limity $n\to\infty$ proves the monotonicity of $\uM_1^{\reg}$. The monotonicity of $\bM_1^{\reg}$ follows from similar arguments.

The weak additivity follows from the fact that the limit in~\eqref{regm} exists. It is therefore left to prove the optimality. 
For this purpose, observe that the optimality of $\uM_1$ and $\bM_1$ implies that for any $\rho\in\mr(A)$ and any $n\in\mbb{N}$ we get that
\be
\frac{1}{n}\uM_1(\rho^{\otimes n})\leq \frac{1}{n}M(\rho^{\otimes n})\leq \frac{1}{n}\bM_1(\rho^{\otimes n})\;.
\ee
But since $M$ is weakly additive, we get 
\be
\frac{1}{n}\uM_1(\rho^{\otimes n})\leq M(\rho)\leq \frac{1}{n}\bM_1(\rho^{\otimes n})\;.
\ee
Hence, taking the limit $n\to\infty$ proves the optimality property.
\end{proof}

\section{Proof of Theorem~\ref{thmmm}}

Here we prove a slightly more detailed version of Theorem~\ref{thmmm}. We start with the following definition.

\begin{definition}
Let $\D:\bigcup_A\md(A)\times\md(A)\to \mbb{R}_+$ be a divergence (see Definition~\ref{qdefd}).
For any $\rho,\sigma\in\md(A)$ we denote its corresponding min and max divergences, respectively, by
\begin{align*}
&\D_{\min}\left(\rho\|\sigma\right)\eqdef \D\left(\begin{bmatrix} 1 & 0\\ 0 & 0\end{bmatrix}\Big\| \bb 2^{-D_{\min}(\rho\|\sigma)} & 0\\ 0 & 1-2^{-D_{\min}(\rho\|\sigma)}\eb\right)\\
&\D_{\max}\left(\rho\|\sigma\right)\eqdef \D\left(\begin{bmatrix} 1 & 0\\ 0 & 0\end{bmatrix}\Big\| \bb 2^{-D_{\max}(\rho\|\sigma)} & 0\\0 & 1-2^{-D_{\max}(\rho\|\sigma)}\eb\right)\;.
\end{align*}
\end{definition}
\begin{remark}
Note that if $\D$ is a relative entropy then from Lemma~5 in~\cite{GT2020a} it follows that $\D_{\max}=D_{\max}$ and $\D_{\min}=D_{\min}$. 
\end{remark}

\begin{theorem*}[A more detailed version of Theorem~\ref{thmmm}] 
Let $\D:\md(A)\times\md(A)\to \mbb{R}_+$ be a quantum divergence as defined in Definition~\ref{qdefd}.
Then, $\D_{\max}$ and $\D_{\min}$ are also divergences, and furthermore,
\be
\D_{\min}\left(\rho\|\sigma\right)\leq \D\left(\rho\|\sigma\right)\leq \D_{\max}\left(\rho\|\sigma\right)\;.
\ee
In particular, if $\D$ is a relative entropy then
for all $\rho,\sigma\in\md(A)$
\be
D_{\min}(\rho\|\sigma)\leq \D(\rho\|\sigma)\leq D_{\max}(\rho\|\sigma)\;.
\ee
\end{theorem*}
\begin{proof}
To show that $\D_{\max}$ and $\D_{\min}$ satisfy the DPI,
observe first that for any two binary probability distributions $(p,1-p)$ and $(q,1-q)$
there exists a classical channel $\mC\in\cptp(X\to X)$ satisfying
\be\label{cc}
\mC(|0\lr 0|)=|0\lr 0|\quad\text{and}\quad\mC\left(p|0\lr 0|+(1-p)|1\lr 1|\right)=q|0\lr 0|+(1-q)|1\lr 1|
\ee
if and only if $p\leq q$. By definition, for any channel $\mE\in\cptp(A\to B)$
\ba
&\D_{\max}\left(\mE(\rho)\big\|\mE(\sigma)\right)= \D\left(|0\lr 0|\;\Big\|\; 2^{-D_{\max}(\mE(\rho)\|\mE(\sigma))}|0\lr 0|+\left(1-2^{-D_{\max}(\mE(\rho)\|\mE(\sigma))}\right)|1\lr 1|\right)
\ea
and also $2^{-D_{\max}(\mE(\rho)\|\mE(\sigma))}\geq 2^{-D_{\max}(\rho\|\sigma)}$. This means that there exists a classical channel $\mC\in\cptp(X\to X)$ satisfying~\eqref{cc} with $q=2^{-D_{\max}(\mE(\rho)\|\mE(\sigma))}$ and $p=2^{-D_{\max}(\rho\|\sigma)}$. Hence, with this classical channel $\mC$ we get
\ba
\D_{\max}\left(\mE(\rho)\|\mE(\sigma)\right)&= \D\Big(\mC(|0\lr 0|)\;\Big\|\; \mC\left(2^{-D_{\max}(\rho\|\sigma)}|0\lr 0|+\big(1-2^{-D_{\max}(\rho\|\sigma)}\big)|1\lr 1|\right)\Big)\\
&\leq \D\Big(|0\lr 0|\;\Big\|\; 2^{-D_{\max}(\rho\|\sigma)}|0\lr 0|+\left(1-2^{-D_{\max}(\rho\|\sigma)}\right)|1\lr 1|\Big)\\
&=\D_{\max}\left(\rho\|\sigma\right)\;.
\ea
Following the same lines as above, one can prove that also $\D_{\min}$ satisfies the DPI. We are now ready to prove the two bounds.

Let $\rho\in\md(A)$, and $\Pi_{\rho}$ denotes the projector to the support of $\rho$. Define the channel (in fact POVM) $\mE\in\cptp(A\to X)$ with $|X|=2$ as
\be
\mE(\sigma)\eqdef\tr\big[\sigma\Pi_\rho\big]|0\lr 0|^X+\tr\big[\sigma\left(I-\Pi_\rho\right)\big]|1\lr 1|^X\;.
\ee
Then,
\ba\label{substi}
\D(\rho\|\sigma)&\geq \D\big(\mE(\rho)\|\mE(\sigma)\big)\\
&=\D\left(|0\lr 0|\Big\|\tr\big[\sigma\Pi_\rho\big]|0\lr 0|+\tr\big[\sigma\left(I-\Pi_\rho\right)\big]|1\lr 1|\right)\;.
\ea
Therefore,  this  gives $\D(\rho\|\sigma)\geq \D_{\min}(\rho\|\sigma)$.

For the second inequality, denote by $t=2^{D_{\max}(\rho\|\sigma)}$, and note that in particular, $t\sigma\geq\rho$ (i.e. $t\sigma-\rho$ is a CP map). Define a channel $\mE\in\cptp(X\to A)$ with $|X|=2$ by 
\be
\mE(|0\lr 0|)=\rho\quad\text{and}\quad\mE(|1\lr 1|)=\frac{1}{t-1}(t\sigma-\rho)\;.
\ee
Furthermore, denote 
\be
\q^X\eqdef\frac{1}{t}|0\lr 0|^X+\frac{t-1}{t}|1\lr 1|^X\;,
\ee
 and observe that $\mE(\q^X)=\sigma$.
Hence,
\ba
\D(\rho\|\sigma)&=\D\Big(\mE(|0\lr 0|^X)\big\|\mE(\q^X)\Big)\\
&\leq \D\left(|0\lr 0|^X\big\|\q^X \right)\\
&=\D_{\max}(\rho\|\sigma)\;.
\ea
This completes the proof.
\end{proof}

\section{Continuity of quantum relative entropies}

In Lemma 5 of~\cite{GT2020a} we showed that if $\D$ is a relative entropy (as defined in Definition~\ref{qdefd} of the main text), then
\begin{align}
	\D \bigg( \begin{bmatrix} 1 & 0 \\ 0 & 0 \end{bmatrix} \bigg\| \begin{bmatrix} 1-\eps & 0 \\ 0 & \eps \end{bmatrix} \bigg) = - \log (1-\eps) \,. \label{eq1}
\end{align}
We use this property to prove the following theorem.

\begin{theorem*}
	Let $\rho,\sigma,\omega \in \md(A)$ be three quantum states. Then,
	\begin{align}
		 \D(\rho \| \sigma) - \D(\rho \| \omega) \leq 
		D_{\max}(\omega\|\sigma) \;,
	\end{align}
	with the convention $\pm\infty\leq+\infty$. Moreover, if $\lambda_{\min}(\omega)> \| \sigma - \omega \|_{\infty}$ then
	\begin{equation}
	D_{\max}(\omega\|\sigma)\leq - \log \left( 1 - \frac{ \| \sigma - \omega \|_{\infty} }{ \lambda_{\min}(\omega) } \right)
	\end{equation}
\end{theorem*}

\begin{proof}
	For $|A|=1$ the statement is trivial so we can assume $|A| \geq 2$. For any $\epsilon>0$, we may write
	\begin{align}
		\sigma = (1-\eps) \omega + \eps \tau , \quad \textrm{where} \quad \tau = \omega + \frac{1}{\eps} (\sigma - \omega) \,.
	\end{align}
	Note that $\tau\geq 0$ iff $\sigma \geq (1-\eps) \omega$ or equivalently iff $\epsilon\geq 1-2^{-D_{\max}
	(\omega\|\sigma)}$.
	 Therefore, $\tau \in \md(A)$ for $\epsilon=1-2^{-D_{\max}
	(\omega\|\sigma)}$.
	Now, using Eq.~\eqref{eq1}, additivity and data-processing, we get
	\begin{align}
		\D(\rho\|\omega) - \log ( 1 - \eps ) = \D \bigg( \rho \otimes \begin{bmatrix} 1 & 0 \\ 0 & 0 \end{bmatrix} \bigg\| \omega \otimes \begin{bmatrix} 1-\eps & 0 \\ 0 & \eps \end{bmatrix} \bigg) \geq \D(\rho\|\sigma)
	\end{align}
	The inequality is the DPI with a map that acts as an identity upon measuring $[1, 0]$ in the second register, and produces a constant output $\tau$ upon measuring $[0,1]$ in the second register.
\end{proof}

So, in particular, the function $\sigma \mapsto \D(\rho \| \sigma)$ is continuous on the set of density matrices in $\md(A)$ with full support. We now show that also the function $\rho \mapsto \D(\rho \| \sigma)$ is continuous.

\begin{theorem*}
	Let $\rho, \omega, \sigma \in \md(A)$ be quantum states. Then, we have
	\begin{align}
		\D(\rho\|\sigma) - \D(\omega\|\sigma) &\leq \min_{0\leq s\leq 2^{-D_{\max}(\omega\|\rho)}} D_{\max}\left(\rho+s(\sigma-\omega)\big\|\sigma\right)\\
		&\leq \log \left( 1 + \frac{  \| \rho - \omega \|_{\infty} } {\lambda_{\min}(\omega) \lambda_{\min}(\sigma) }  \right)
		\end{align}
		where the second inequality holds if $\sigma>0$ and $\lambda_{\min}(\omega)>\|\rho-\omega\|_\infty$.
		\end{theorem*}

\begin{proof}
	In somewhat of a variation of the previous theorem, consider the linear map (here $\eps\in [0,1]$)
	\begin{align}
		\mathcal{E}: X \mapsto (1-\eps) X + \eps \tau \quad \textrm{where} \quad \tau = \omega + \frac{1}{\eps} (\rho - \omega) ,
	\end{align}
	Again, the condition $\eps\geq1-2^{-D_{\max}
	(\omega\|\rho)}$ is equivalent to $\tau\geq 0$ which ensures that $\mathcal{E}$ is CPTP.
	Clearly then $\mathcal{E}(\omega) = \rho$. Moreover, we would like that
	\begin{align}
		(1-\nu) \mathcal{E}(\sigma) + \nu \kappa = \sigma ,
	\end{align}
	where $\nu \in (0,1)$ and $\kappa \in \md(A)$ is a state still to be defined. Solving for $\kappa$ yields
	\begin{align}
		\kappa = \sigma + \left( \frac{\eps}{\nu} - \eps \right) (\sigma - \tau)  
	\end{align}
	which is positive semi-definite iff
	\begin{equation}\label{c1}
	(1-\nu)^{-1}\sigma\geq (1-\epsilon)\sigma+\rho-(1-\epsilon)\omega\;.
	\end{equation}
	 We can think of $\kappa$ fixing the damage done by applying $\mathcal{E}$ on $\sigma$.
	Now we again use Eq.~\eqref{eq1}, additivity and DPI to find
	\begin{align}
		\D(\omega\|\sigma) - \log ( 1 - \nu ) = \D \bigg( \rho \otimes \begin{bmatrix} 1 & 0 \\ 0 & 0 \end{bmatrix} \bigg\| \omega \otimes \begin{bmatrix} 1-\nu & 0 \\ 0 & \nu \end{bmatrix} \bigg) \geq \D(\rho\|\sigma)
	\end{align}
	where we use a channel that acts as $\mathcal{E}$ when measuring $[1,0]$ in the second register and outputs $\kappa$ when measuring $[0,1]$. Therefore, taking the smallest possible value of $t:=(1-\nu)^{-1}$ under the constraint~\eqref{c1} and the condition that $s:=1-\epsilon\leq 2^{-D_{\max}(\omega\|\rho)}$ gives
	\begin{align}
	\D(\rho\|\sigma)-\D(\omega\|\sigma)&\leq\log\min\Big\{t\geq 0\;:\;(t-s)\sigma\geq\rho-s\omega\geq 0\;,\;s\geq 0\Big\}\\
	&=\min_{0\leq s\leq 2^{-D_{\max}(\omega\|\rho)}} D_{\max}\left(\rho+s(\sigma-\omega)\big\|\sigma\right)
	\end{align}
	If $\mu\eqdef\lambda_{\min}(\sigma)>0$ we can take $t=1+\frac{1-s}{\mu}$. Note that for this choice of $t$ we have
	\be
	(t-s)\sigma=(1-s)(1+\mu)\frac{\sigma}{\mu}\geq (1-s)(1+\mu)I^A\geq \rho-s\omega
	\ee
	since $\rho-s\omega$ is a subnormalized state with trace $1-s$. Moreover, if $\lambda_{\min}(\omega)\geq\|\rho-\omega\|_{\infty}$ then we can take $s=1-\frac{\|\rho-\omega\|_{\infty}}{\lambda_{\min}(\omega)}$ since in this case $s\leq 2^{-D_{\max}(\omega\|\rho)}$ (or equivalently $\rho\geq s\omega$). 	We therefore get for these choices of $t$ and $s$
	\begin{align}
	\D(\rho\|\sigma)-\D(\omega\|\sigma)\leq \log t=\log \left( 1 + \frac{  \| \rho - \omega \|_{\infty} } {\lambda_{\min}(\omega) \lambda_{\min}(\sigma) }  \right)	\;.\end{align}
	This completes the proof.
	\end{proof}

\section{Faithfulness of quantum divergences}
\begin{definition*}
A quantum divergence, $\D$, is said to be \emph{faithful} if for any $\rho,\sigma\in\md(A)$, the condition $\D(\rho\|\sigma)=0$ implies $\rho=\sigma$. 
\end{definition*}
\begin{theorem*}
Let $\D$ be a quantum divergence. Then, $\D$ is faithful if and only if its reduction to classical (diagonal) states is faithful.
\end{theorem*}
\begin{proof}
Clearly, if $\D$ is faithful on quantum states it is also faithful on classical states as the latter is a subset of the former. Suppose now that $\D$ is faithful on classical states, and suppose by contradiction that there exists $\rho\neq\sigma\in\md(A)$ such that $\D(\rho\|\sigma)=0$. Then, there exists a basis of $A$ such that the diagonal of $\rho$ in this basis does not equal to the diagonal of $\sigma$. Let $\Delta\in\cptp(A\to A)$ be the completely dephasing channel in this basis. Then, $\Delta(\rho)\neq\Delta(\sigma)$ and we get
\be
\D\big(\Delta(\rho)\big\|\Delta(\sigma)\big)\leq\D(\rho\|\sigma)=0\;.
\ee
But since $\D$ is faithful on diagonal states we get the contradiction that $\Delta(\rho)=\Delta(\sigma)$.
Hence, $\D$ is faithful also on quantum states.
\end{proof} 

When combining the above lemma with the condition on faithfulness given in Theorem~17 of~\cite{GT2020a}  we get that almost all quantum relative entropies are faithful.

\textbf{Corollary.}
{\it Let $\D$ be a quantum divergence (not necessarily additive) that on classical systems reduces to a classical relative entropy (i.e. additive classical divergence). Then, $\D$ is faithful if and only if there exists classical system $X$ and  two classical states $\p,\q\in\md(X)$ with the same support such that
\be
\D(\p\|\q)\neq 0\;.
\ee}

The corollary above (cf. Theorem 17 in~\cite{GT2020a}) demonstrates that if $\D$ is \emph{not} faithful then it must be zero on all classical states with the same support.
$D_{\min}$ is an example of such non-faithful divergence. More details on non-faithful classical relative entropies can be found in~\cite{GT2020a}.

\section{Proof of Theorem~\ref{thmmax} and the maximal extension of R\'enyi divergences}

The maximal extension in~\eqref{mme} of a divergence $\D_1$ can be expressed as
\be\label{simom}
\bbd_1(\rho\|\sigma)=\inf_{|X|\in\mbb{N}}\Big\{\D_1(\p\|\q)\;:\;\p,\q\in\md(X)\;\;,\;\; \rho=\sum p_x\omega_x\;\;,\;\;\sigma=\sum q_x\omega_x\;\;,\;\;\{\omega_x\}\subset\md(A)\Big\}\;.
\ee
We use this expression to prove Theorem~\ref{thmmax}.

\begin{theorem*}[Theorem~\ref{thmmax} of the main text]
Let $\D_1:\md(X)\times\md(X)\to\mbb{R}$ be a classical divergence, and let $\bbd_1:\md(A)\times\md(A)\to\mbb{R}$ be its maximal  extension to quantum states. Then, for any pure state $\psi\eqdef |\psi\lr\psi|\in\md(A)$ and any mixed state $\sigma\in\md(A)$ we have
\be
{\bf \bD}_1(\psi\big\|\sigma)=D_{\max}(\psi\big\|\sigma)=\log\la\psi|\sigma^{-1}|\psi\ra\;.
\ee 
\end{theorem*}
\begin{remark}
The theorem above implies that the maximal extension $\bD_\alpha$ of the (classical) R\'enyi divergence cannot be equal to the Petz quantum R\'enyi divergence. Specifically,
\be
D_{\alpha}^{\text{Petz}}(\psi\|\sigma)=\frac{1}{\alpha-1}\log\la\psi|\sigma^{1-\alpha}|\psi\ra\;,
\ee
which in general for $\alpha<2$ is different then $\bD_\alpha(\psi\|\sigma)=D_{\max}(\psi\|\sigma)$, unless $\psi$ and $\sigma$ commutes. On the other hand, for $\alpha=2$, as we will see below, $\bD_{\alpha=2}(\rho\|\sigma)=D_{\alpha=2}^{\rm Petz}(\rho\|\sigma)$ for all $\rho,\sigma\in\md(A)$.
\end{remark}

\begin{proof}
Since $\psi$ is pure, the condition $\psi=\sum_{x}p_{x}\omega_{x}$, can hold only if for all $x$ such that $p_x\neq 0$ we have $\omega_x=\psi$. W.l.o.g. let the $k$ first components of $\p$ be non-zero, while all the remaining components are zero. This implies that the second condition can be expressed as
\be
\sigma=\sum_{x=1}^{k}q_x\psi+\sum_{x=k+1}^{n}q_{x}\omega_{x}\;.
\ee
Denote by $s\eqdef\sum_{x=1}^{k}q_x$, and observe that there exists such $\{\omega_{x}\}_{x=k+1}^{n}$ if and only if 
\be
\sigma\geq s\psi
\ee
or in other words, iff $s^{-1}\geq 2^{D_{\max}(\psi\|\sigma)}$. Consider the classical channel $\mC\in\cptp(X\to X)$ defined by
\be
\mC(|x\lr x|)=|1\lr 1|\;\;\forall x=1,...,k\quad\text{and}\quad\mC(|x\lr x|)=|2\lr 2|\;\;\forall x=k+1,...,n\;.
\ee
Therefore, we must have $\D_1(\p,\q)\geq \D_1\big(\mC(\p),\mC(\q)\big)=\D_1\Big(|1\lr 1|\;,\;s|1\lr 1|+(1-s)|2\lr 2|\Big)$.
This means that w.l.o.g. we can assume that $\p=|1\lr 1|$ and $\q$ is binary; i.e. $\q=s|1\lr 1|+(1-s)|2\lr 2|$ so that 
$\D_1(\p\|\q)=-\log(s)$ (cf.~\eqref{160}). But since we must have $s^{-1}\geq 2^{D_{\max}(\psi\|\sigma)}$, the minimum value is achieved when $s^{-1}= 2^{D_{\max}(\psi\|\sigma)}$. That is, $\D_1(\p\|\q)=D_{\max}(\psi\|\sigma)$. This completes the proof.
\end{proof}

Note that if in the expression~\eqref{simom} of the maximal extension, we denote $r_x=p_x/q_x$ and $E_x\eqdef q_x\sigma^{-1/2}\omega_x\sigma^{-1/2}$, then we get that $\rho$ and $\sigma$ in~\eqref{simom} satisfy
\be
\sigma^{-1/2}\rho\sigma^{-1/2}=\sum_xr_xE_x\quad\quad I=\sum_{x}E_x
\ee
We can therefore express
\be\label{optim}
\bbd_1(\rho\|\sigma)=\inf\Big\{\D_1(\q\circ\r\|\q)\;:\;\sigma^{-\frac12}\rho\sigma^{-\frac12}=\sum_xr_xE_x\;\;,\;\;E_x\geq 0\;\;,\;\;\sum_{x}E_x=I\;\;,\;\;q_x\eqdef\tr[E_x\sigma]\Big\}
\ee
where $\q\circ\r\eqdef(q_1r_1,...,q_nr_n)^T$. Note that the vector $\q\circ\r$ is a probability vector since the constriant
$\sigma^{-\frac12}\rho\sigma^{-\frac12}=\sum_xr_xE_x$ gives
\be
1=\tr[\rho]=\sum_xr_x\tr[\sigma E_x]=\sum_{x}r_xq_x\;.
\ee
In~\cite{Matsumoto2018}, Matsumoto used the above expression of $\bbd$ to show that for certain $f$-divergences, the optimal choice of $\r$ and $E_x$ is to take $E_x=|\psi_x\lr\psi_x|$ where $\{|\psi_x\ra\}$ and $\{r_x\}$ are the eigenvectors and eigenvalues of 
$\sigma^{-1/2}\rho\sigma^{-1/2}$. In particular, this is the optimal choice for R\'enyi entropies with $\alpha\in(0,2]$ so that in this case
\ba
\bD_\alpha(\rho\|\sigma)&=D_\alpha(\q\circ\r\|\q)\\
&=\frac{1}{\alpha-1}\log\sum_{x}(q_xr_x)^\alpha q_{x}^{1-\alpha}=\frac{1}{\alpha-1}\log\sum_{x}q_xr_x^\alpha \\
&=\frac{1}{\alpha-1}\log\sum_{x}\la\psi_x|\sigma|\psi_x\ra\la\psi_x|\left(\sigma^{-\frac12}\rho\sigma^{-\frac12}\right)^\alpha|\psi
_x\ra\\
&=\frac{1}{\alpha-1}\log\tr\left[\sigma\left(\sigma^{-\frac12}\rho\sigma^{-\frac12}\right)^\alpha\right]\;.
\ea
Note that for $\alpha=2$ we get $\bD_2(\rho\|\sigma)=D_{2}^{\rm Petz}(\rho\|\sigma)$. Remarkably, the above closed formula for $\bD_\alpha$ demonstrates that $\bD_\alpha$ is additive under tensor product and hence it is a relative entropy. However, the above formula only holds for $\alpha\in(0,2]$. For $\alpha>2$ the optimizer in~\eqref{optim} is not given by the eigenvalues and eigenvectors of $\sigma^{-1/2}\rho\sigma^{-1/2}$. Instead, another set of $\{r_x,\;E_x\}$ is the optimizer of~\eqref{optim}. We therefore conclude that for $\alpha>2$ we have 
\be
\bD_\alpha(\rho\|\sigma)\leq \frac{1}{\alpha-1}\tr\left[\sigma\left(\sigma^{-\frac12}\rho\sigma^{-\frac12}\right)^\alpha\right]
\ee
where the inequality is strict for at least some choices of $\rho$ and $\sigma$. 

\section{Uniqueness of the Umegaki  relative entropy}

We show here that the framework for extensions developed here can be used to single out the Umegaki relative entropy as the \emph{only} relative entropy that is asymptotically continuous. This result was first proven in~\cite{Matsumoto2018b} and we provide here an alternative proof.
We say that a relative entropy $\D$ is asymptotically continuous if there exists a continuous function $f:[0,1]\to\mbb{R}_+$ such that $f(0)=0$ and for all $\rho,\rho',\sigma\in\md(A)$, with $\supp(\rho)\subseteq\supp(\sigma)$ and $\supp(\rho')\subseteq\supp(\sigma)$
\be\label{asy}
\left|\D(\rho\|\sigma)-\D(\rho'\|\sigma)\right|\leq f(\epsilon)\log\|\sigma^{-1}\|_1
\ee
where $\epsilon\eqdef\frac{1}{2}\|\rho-\rho'\|_1$. We emphasize that $f$ is independent of $|A|$.

\begin{theorem}\label{uniqueness}
Let $\D$ be a relative entropy that is asymptotically continuous. Then, for all $A$ and all $\rho,\sigma\in\md(A)$,
\be
\D(\rho\|\sigma)=D(\rho\|\sigma)\eqdef\tr[\rho\log\rho]-\tr[\rho\log\sigma]\;.
\ee
\end{theorem}

\begin{remark}
In~\cite{WGE2017} the uniqueness of the quantum relative entropy was established for a slightly different approach in which the asymptotic continuity~\eqref{asy} is replaced with continuity in the first argument, and in addition super-additivity is assumed.
Characterization of the von Neumann entropy in terms of correlated catalysts was also studied in~\cite{Muller2018,BEG2019} (see also~\cite{RW2019} for characterization of the Kullback–Leibler divergence in terms of a type of relative majorization).
Earlier approaches based on unique measure of volume were studied in~\cite{Hall1999,Hall2000}.
\end{remark}

The asymptotic continuity of the Umegaki relative entropy can be characterized as
\be
\left|D(\rho\|\sigma)-D(\rho'\|\sigma)\right|\leq\epsilon\log\|\sigma^{-1}\|_{\infty}+(1+\epsilon)h\left(\frac{\epsilon}{1+\epsilon}\right)
\quad\quad\forall\;\rho,\rho',\sigma\in\md(A)\;,
\ee
where $\epsilon\eqdef\frac{1}{2}\|\rho-\rho'\|_1$ and $h(x)\eqdef-x\log (x)-(1-x)\log(1-x)$ is the binary Shannon entropy. 

\begin{lemma*}
Let $\rho,\sigma\in\md(A)$ with $\supp(\rho)\subseteq\supp(\sigma)$ and let $\D$ be a quantum relative entropy satisfying~\eqref{asy} (i.e. $\D$ is asymptotically continuous).
Then,
\ba\label{gg}
\D(\rho\|\sigma)&=\lim_{\epsilon\to 0^+}\liminf_{n\to\infty}\inf_{\rho_{n}'\in\mb_\epsilon\left(\rho^{\otimes n}\right)}\frac{1}{n}\D\left(\rho_{n}'\big\|\sigma^{\otimes n}\right)\\
&=\lim_{\epsilon\to 0^+}\limsup_{n\to\infty}\sup_{\rho_{n}'\in\mb_\epsilon\left(\rho^{\otimes n}\right)}\frac{1}{n}\D\left(\rho_{n}'\big\|\sigma^{\otimes n}\right)
\ea
where for any system $A$ and any $\omega\in\md(A)$
\be
\mb_\epsilon\left(\omega\right)\eqdef\Big\{\sigma\in\md(A)\;:\;\frac12\|\omega-\sigma\|_1\leq\epsilon\Big\}
\ee
\end{lemma*}
\begin{proof}
Applying~\eqref{asy} to $n$ copies with $\frac{1}{2}\|\rho_{n}'-\rho^{\otimes n}\|_1\leq\epsilon$  for some fixed $\epsilon>0$ gives
\be
\left|\D(\rho\|\sigma)-\frac{1}{n}\D(\rho_{n}'\|\sigma^{\otimes n})\right|\leq f(\epsilon)\log\|\sigma^{-1}\|_{\infty}
\ee
Therefore, by taking the $\liminf_{n\to\infty}$ or $\limsup_{n\to\infty}$ on both sides of the equation above followed by $\lim_{\epsilon\to 0^+}$ completes the proof.
\end{proof}

\begin{theorem*}
The Umegaki relative entropy is the only quantum relative entropy satisfying~\eqref{gg}.
\end{theorem*}
\begin{remark}
Note that the theorem above implies Theorem~\ref{uniqueness} of the main text since the lemma above states that~\eqref{asy} implies~\eqref{gg}. Therefore, the Umegaki relative entropy is the only asymptotically continuous relative entropy.
\end{remark}

\begin{proof}
Let $\D(\rho\|\sigma)$ be a divergence satisfying~\eqref{gg}. Therefore, 
\ba
\D(\rho\|\sigma)&=\lim_{\epsilon\to 0^+}\liminf_{n\to\infty}\inf_{\rho_{n}'\in\mb_\epsilon\left(\rho^{\otimes n}\right)}\frac{1}{n}\D\left(\rho_{n}'\big\|\sigma^{\otimes n}\right)\\
&\leq\lim_{\epsilon\to 0^+}\liminf_{n\to\infty}\inf_{\rho_{n}'\in\mb_\epsilon\left(\rho^{\otimes n}\right)}\frac{1}{n}D_{\max}\left(\rho_{n}'\big\|\sigma^{\otimes n}\right)\\
&=\lim_{\epsilon\to 0^+}\liminf_{n\to\infty}\frac{1}{n}D_{\max}^{\epsilon}\left(\rho^{\otimes n}\big\|\sigma^{\otimes n}\right)\\
&=D(\rho\|\sigma)\;,
\ea
where the inequality follows from~\eqref{32}, and the last equality from the asymptotic equipartition property. 
Conversely, from the lower bound in~\eqref{32}
\ba
\D(\rho\|\sigma)&=\lim_{\epsilon\to 0^+}\limsup_{n\to\infty}\sup_{\rho_{n}'\in\mb_\epsilon\left(\rho^{\otimes n}\right)}\frac{1}{n}\D\left(\rho_{n}'\big\|\sigma^{\otimes n}\right)\\
&\geq\lim_{\epsilon\to 0^+}\limsup_{n\to\infty}\sup_{\rho_{n}'\in\mb_\epsilon\left(\rho^{\otimes n}\right)}\frac{1}{n}D_{\min}\left(\rho_{n}'\big\|\sigma^{\otimes n}\right)\\
&\geq D(\rho\|\sigma)
\ea
where the last line follows from the lemma below. This completes the proof.
\end{proof}
\begin{lemma*}
Let $\rho,\sigma\in\md(A)$ with $\supp(\rho)\subset\supp(\sigma)$. Define,
\be
D_{\min}^{\epsilon}(\rho\|\sigma)\eqdef\sup_{\rho'\in\mb_\epsilon\left(\rho\right)}D_{\min}\left(\rho'\big\|\sigma\right)
\ee
Then, for any $0<\epsilon<1$,
\be
\lim_{n\to\infty}\frac{1}{n}D_{\min}^{\epsilon}(\rho^{\otimes n}\|\sigma^{\otimes n})\geq D(\rho\|\sigma)\;.
\ee
\end{lemma*}

\begin{proof}
	The proof employs the following auxiliary quantity, related to the information spectrum:
	\begin{align}
	    D_{s}^{\eps}(\rho\|\sigma) := \sup \big\{R \in \mathbb{R} \,\big|\, \tr \big( \rho \{\rho \leq 2^{R} \sigma \} \big) \leq \eps \big\} = \sup \big\{R \in \mathbb{R} \,\big|\, \tr \big( \rho \{\rho > 2^{R} \sigma \} \big) \geq 1- \eps \big\} .
	\end{align}
	It is intimately related to hypothesis testing, e.g.\ we have~\cite[Lemma 12]{tomamichel12}
\begin{align}
    D_{s}^{\eps}(\rho\|\sigma) \leq D_{h}^{\eps}(\rho\|\sigma) 
    \leq D_{s}^{\eps+\delta}(\rho\|\sigma) - \log \delta \label{eq:hypo-info},
\end{align}
where
\be
D_{h}^{\eps}(\rho\|\sigma)\eqdef-\log\min\Big\{\tr[\sigma\Pi]\;:\;0\leq\Pi\leq I\;\;,\;\;\tr[\rho\Pi]\geq 1-\epsilon\Big\}
\ee
is the Hypothesis testing divergence.

	Now take $\lambda = D_s^{\frac{\eps^2}{2}}(\rho\|\sigma)$. Then there exist $\delta > 0$ arbitrarily small, such that
	\begin{align}
		\tr (\rho P) \geq 1-\frac{\eps^2}{2}, \quad \tr (\sigma P) \leq 2^{-\lambda+\delta} \tr (\rho P) \leq 2^{-\lambda+\delta} \quad \textrm{with} \quad P = \big\{ \rho > 2^{\lambda-\frac{\eps^2}{2}} \sigma \big\}
	\end{align}
	We define $\bar{\rho} := \frac{1}{\tr(P\rho)} P \rho P$ and using the gentle measurement lemma and the above we can verify that $\bar{\rho} \in \mb_\epsilon\left(\rho\right)$. Hence, $D_{\min}^{\epsilon}(\rho\|\sigma) \geq - \log \tr \big( \bar{\rho}^0 \sigma\big)$. Now taking advantage of the fact that $\bar{\rho}^0 \leq P$ by definition, we infer that 
		$D_{\min}^{\epsilon}(\rho\|\sigma) \geq - \log \tr \big( P \sigma\big) \geq \lambda - \delta$.
	And since $\delta$ is arbitrarily small, we can use~\eqref{eq:hypo-info} to conclude that
	\begin{align}
	 D_{\min}^{\epsilon}(\rho\|\sigma) \geq D_s^{\frac{\eps^2}{2}}(\rho\|\sigma) \geq D_h^{\frac{\eps^2}{4}}(\rho\|\sigma) - \log \frac{4}{\eps^2}
	\end{align}
	The statement of the lemma now follows from a simple application of the quantum Stein's lemma. 
\end{proof}

\section{Extensions from normalized to subnormalized states}

Sub-normalized states are positive semi-definite matrices with trace less or equal to one. We will denote the set of subnormalized states acting on Hilbert space $A$ by 
\be
\tilde{\md}(A)\eqdef\Big\{\trho\in\pos(A)\;:\;\tr[\trho]\leq 1\Big\}\;.
\ee

One of the properties of quantum channels is that they take normalized states to normalized states. When considering subnormalized states,
all trace non-increasing (TNI) CP maps (including CPTP maps) take sub-normalized states to subnormalized states.
In applications, it is quite often useful to quantify distances between subnormalized states with a function that obeys a monotonicity  property (i.e. data processing inequality) under TNI-CP maps. We start by proving Theorem~\ref{gmain} of the main text, and then discuss its applications.

\begin{theorem*}[detailed version of Theorem~\ref{gmain} in the main text]
Let $\D$ be a quantum divergence and $\bbd$ be its maximal extension to sub-normalized states (see~\eqref{dextd}). For any pair of sub-normalized states $(\trho,\tsigma)\in\mr(A\oplus A)$
\be\label{hsh}
\bbd(\trho\|\tsigma)=\D\Big(\trho\oplus\big(1-\tr[\trho]\big)\;\big\|\;\tsigma\oplus\big(1-\tr[\tsigma]\big)\Big)\;.
\ee
Moreover, the minimal extension of $\D$ to subnormalized states, $\ubd$, satisfies
\be
\ubd(\trho\|\tsigma)=0
\ee
for all subnormalized states $\trho,\tsigma\in\tilde{\md}(A)$ with either $\tr[\trho]<1$ or $\tr[\tsigma]<1$.
\end{theorem*}
\begin{proof}
Let $\rho,\sigma\in\md(R)$ and $\mE\in\cp(R\to A)$ be a TNI-CP map such that $\trho=\mE(\rho)$ and $\tsigma=\mE(\sigma)$. Moreover, define $\mN\in\cptp(R\to A\oplus\mbb{C})$ as
\be
\mN(\omega)\eqdef\mE(\omega)\oplus\big(\tr[\omega]-\tr[\mE(\omega)]\big)\quad\forall\omega\in\ml(A)\;.
\ee
Then, since $\mN$ is a CPTP map,
\ba
\D(\rho\|\sigma)&\geq \D\big(\mN(\rho)\|\mN(\sigma)\big)\\
&=\D\Big(\mE(\rho)\oplus\big(1-\tr[\mE(\rho)]\big)\;\big\|\;\mE(\sigma)\oplus\big(1-\tr[\mE(\sigma)]\big)\Big)\\
&=\D\Big(\trho\oplus\big(1-\tr[\trho]\big)\;\big\|\;\tsigma\oplus\big(1-\tr[\tsigma]\big)\Big)\;.
\ea
Since the above inequality holds for all such $\rho,\sigma,\mE$ we must have that $\bbd(\rho\|\sigma)$ is no smaller than the RHS on~\eqref{g1}.
To prove the converse inequality, take $R=A\oplus\mbb{C}$, $\rho=\trho\oplus\big(1-\tr[\trho]\big)$, $\sigma=\tsigma\oplus\big(1-\tr[\tsigma]\big)$, and $\mE(\cdots)\eqdef P(\cdot)P^\dag$, where $P$ is the projection to $A$ in $R$. Then, $\trho=\mE(\rho)$ and $\tsigma=\mE(\sigma)$
so that by definition (see~\eqref{dextd}) we must have $\bbd(\trho\|\tsigma)\leq \D(\rho\|\sigma)$. This completes the proof of the equality in~\eqref{hsh}.

The minimal extension $\ubd$ can be expressed for any $\trho,\tsigma\in\tilde{\md}(A)$ as
\be
\ubd(\trho\|\tsigma)\eqdef\sup\D(\mE(\trho)\|\mE(\tsigma))
\ee
where the supremum is over all systems $R$ and all $\mE\in\cp(A\to R)$ such that $\mE$ is trace non-increasing and
$\mE(\trho)$ and $\mE(\tsigma)$ are normalized states. However, such $\mE$ does not exists if either $\trho$ or $\tsigma$ has trace strictly smaller than one. This completes the proof.
\end{proof}

\subsection{Applications to Distance Measures}

In this subsection we apply the Theorem~\ref{gmain} to divergences that are also metrics. Explicity, we assume that $\D$ is symmetric, 
\be
\D(\rho\|\sigma)=\D(\sigma\|\rho)\quad\quad\forall\;\rho,\sigma\in\md(A)\;,
\ee 
and satisfies the triangle inequality
\be
\D(\rho\|\sigma)\leq \D(\rho\|\omega)+\D(\omega\|\sigma)\quad\quad\;\forall\;\rho,\sigma,\omega\in\md(A)\;.
\ee
\begin{lemma*}
Let $\D$ be a quantum divergence that is also a metric. Then, its maximal extension to subnormalized states, $\bbd$, is also a metric (that satisfies the DPI under trace non-increasing CP maps).
\end{lemma*}

\begin{proof}
The symmetry property  of $\D$ follows trivially from the symmetry of $\D$ and Theorem~\ref{gmain}.
It is therefore lest to show that for any three subnormalized states $\trho,\tsigma,\tomega\in\umd(A)$
\be
\bbd(\trho\|\tsigma)\leq \bbd(\trho\|\tomega)+\bbd(\tomega\|\tsigma)\;.
\ee
Again, this property follows directly from the closed formula in Theorem~\ref{gmain}, and the triangle inequality of $\D$.
\end{proof}

\subsubsection*{Examples 1: The Generalized Trace Distance}

If we take $\D$ to be the trace distance defined by 
\be
D(\rho,\sigma)\eqdef\frac12\|\rho-\sigma\|_1\quad\quad\forall\;\rho,\sigma\in\md(A)\;,
\ee
then Theorem~\ref{gmain} states that its maximal extension is given by
\be
\bD(\trho,\tsigma)=\frac12\|\trho-\tsigma\|_1+\frac12\big|\tr[\trho-\tsigma]\big|\quad\quad\forall\;\trho,\tsigma\in\tilde{\md}(A)\;.
\ee
This formula was introduced in~\cite{Tomamichel-2012a,Tomamichel2015} and our formalism indicate that it is the largest extension of the trace distance to subnormalized states, meaning that any other extension of the trace distance must be smaller than the generalized trace distance.

\subsubsection*{Examples 2: The Generalized Fidelity}

The Fidelity is a measure defined by
\be
F(\rho,\sigma)\eqdef\big\|\sqrt{\rho}\sqrt{\sigma}\big\|_1\;.
\ee
It satisfies the DPI in the opposite direction. Therefore, when applying Theorem~\ref{gmain} to the fidelity we get that its minimal extension to sub-normalized states is given by
\be
\bar{F}(\trho,\tsigma)=\big\|\sqrt{\trho}\sqrt{\tsigma}\big\|_1+\sqrt{(1-\tr[\trho])(1-\tr[\tsigma])}\quad\quad\forall\;\trho,\tsigma\in\tilde{\md}(A)\;.
\ee
This measure was introduced in~\cite{Tomamichel-2012a,Tomamichel2015}.
Theorem~\ref{gmain} also implies that any other fidelity-type measure on subnormalized states, that reduces to $F$ on normalized states, must be no smaller than the expression above. This provide a strong motivation to use the generalized fidelity in applications,
since if the generalized fidelity is close to one, it means that any other extension of the fidelity that satisfy the DPI (in the opposite direction) must be close to one.

\subsection{Relative entropies of subnormalized states}

The Umegaki relative entropy of subnormalized states is defined for any $\trho,\tsigma\in\tilde{\md}(A)$ as
\ba\label{asdf}
\bD(\trho\|\tsigma)&\eqdef D\Big(\trho\oplus\big(1-\tr[\trho]\big)\big\|\tsigma\oplus\big(1-\tr[\tsigma]\big)\Big)\\
&=D(\trho\|\tsigma)+\big(1-\tr[\trho]\big)\log\frac{1-\tr[\trho]}{1-\tr[\tsigma]}\;.
\ea
The extended relative entropy $\bD$ satisfies the following two key properties:
\begin{enumerate}
\item Faithfulness. For any $\trho,\tsigma\in\umd(A)$, $\bD(\trho\|\tsigma)=0$ iff $\trho=\tsigma$.
\item Data Processing Inequality. For any $\trho,\tsigma\in\tilde{\md}(A)$ and a TNI map $\mE\in\cp(A\to B)$ we have \be\bD(\mE(\trho)\|\mE(\tsigma))\leq \bD(\trho\|\tsigma)\;.\ee
\end{enumerate}
Note in particular that $\bD$ behaves monotonically not only under CPTP maps but also under TNI-CP maps.
Moreover, $\bD$ is always non-negative. However, in general, the additivity property of the relative entropy does no carry over to $\bD$. Note also that if $\tr[\trho]<1$ and $\tr[\tsigma]<1$ then
\be
\lim_{n\to\infty}\frac1n\bD\left(\trho^{\otimes n}\big\|\tsigma^{\otimes n}\right)=0\;.
\ee
This means that we cannot use the techniques discussed earlier to define (weakly) additive quantities. 
\begin{corollary}
Let $\mE\in\cp(A\to B)$ be a TNI CP map, $\rho,\sigma\in\md(A)$ be normalized states, and $D$ be the Umegaki relative entropy. Then,
\be
D\left(\mE(\rho)\big\|\mE(\sigma)\right)\leq D(\rho\|\sigma)-\big(1-\tr[\mE(\rho)]\big)\log\frac{1-\tr[\mE(\rho)]}{1-\tr[\mE(\sigma)]}
\ee
\end{corollary}
\begin{proof}
Follows trivially from~\eqref{asdf} by taking $\trho=\mE(\rho)$ and $\tsigma=\mE(\sigma)$.
\end{proof}

One can define the extension of R\'enyi divergences to subnormalized states in a similar way. Particularly interesting is the extension of $D_{\max}$ which takes the form
\ba
\bD_{\max}(\trho\|\tsigma)&\eqdef D_{\max}\Big(\trho\oplus\big(1-\tr[\trho]\big)\big\|\tsigma\oplus\big(1-\tr[\tsigma]\big)\Big)\\
&=\log\max\left\{2^{D_{\max}(\trho\|\tsigma)}\;,\;\frac{1-\tr[\trho]}{1-\tr[\tsigma]}\right\}\;.
\ea
In particular, for $\tr[\tsigma]\leq\tr[\trho]$ we have $\bD_{\max}(\trho\|\tsigma)=D_{\max}(\trho\|\tsigma)$.

\section{Extensions from pure states to mixed states}

Here we consider extensions from a QRT that is defined on pure states and extend it to a QRT that is defined on mixed states.
Let $\mf$ denotes a QRT (the domain is $\md(A)$ for any physical system $A$).
Denoting by $\pure(A)$ the subset of all pure states in $\md(A)$.
Let $M_1:\bigcup_{A}\pure(A)\to\mbb{R}$ be a resource measure defined on pure states. We assume that $M_1$ is non-increasing under $\mf$. Note that a channel $\mE\in\mf(A\to B)$ can take a pure state $\psi\in\pure(A)$ to a mixed state $\mE(\psi)\in\md(A)$. Therefore, the relation $M_1(\mE(\psi))\leq M_1(\psi)$ holds only when $\mE(\psi)$ is a pure state in $\pure(B)$.

The minimal and maximal extensions of $M_1$ are given by
\be\label{909}
\overline{M}_1(\rho)\eqdef\inf\Big\{M_1(\psi)\;:\;\rho=\mE(\psi)\;,\;\mE\in\mf(B\to A)\;,\;\psi\in\pure(B)\Big\}
\ee
and its minimal extension as
\be\label{908}
\underline{M}_1(\rho)\eqdef\sup\Big\{M_1\big(\mE(\rho)\big)\;:\;\mE\in\mf(A\to B)\;,\;\mE(\rho)\in\pure(B)\Big\}
\ee
We first show that the expression for the maximal extension can be simplified when the QRT is purifiable.
\begin{definition*}
A QRT $\mf$ is said to be \emph{purifiable} if for any two systems $A$ and $B$, and any free channel
$\mE\in\mf(A\to B)$, there exists a system $E$ and an isometry $\mV\in\mf(A\to BE)$ such that 
\be\label{isoe}
\mE^{A\to B}=\tr_E\circ\mV^{A\to BE}\;.
\ee
\end{definition*} 

\begin{theorem*}
Let $\mf$  be a purifiable QRT, and let $M_1:\bigcup_{A}\pure(A)\to\mbb{R}$ be a resource measure defined on pure states. Then, the maximal extension $\overline{M}_1$ of $M_1$ to mixed states is given for all $\rho\in\md(A)$ by
\be
\overline{M}_1(\rho^A)=\inf\Big\{M_1(\psi^{RA})\;:\;\tr_{R}\left[\psi^{RA}\right]=\rho^A\;\;,\;\;\psi\in\pure(RA)\Big\}
\ee
\end{theorem*}
\begin{proof}
Let $\psi\in\pure(RA)$ be a purification of $\rho^A$. Take the system $B$ in~\eqref{909} to be the composite system $RA$, and  $\mE\in\mf(B\to A)$ be $\mE^{RA\to A}=\tr_{R}$ we get that $\mE(\psi^{RA})=\rho^A$. Hence, by definition~\eqref{909} we get that
\be
\overline{M}_1(\rho)\leq M_1(\psi^{RA})\;.
\ee
Since the above inequality holds for any purification $\psi^{RA}$ of $\rho^A$, it also holds for the infimum over such purifications. We therefore get 
\be
\overline{M}_1(\rho^A)\leq\inf\Big\{M_1(\psi^{RA})\;:\;\tr_{R}\left[\psi^{RA}\right]=\rho^A\;\;,\;\;\psi\in\pure(RA)\Big\}\;.
\ee
For the other direction, let $\rho^A=\mE(\psi^B)$, where $\mE\in\mf(B\to A)$ and $\psi\in\pure(B)$. Let $\mV\in\mf(B\to RA)$ be an isometry purifying $\mE^{B\to A}$. Therefore, $\rho^A=\tr_R\left[\mV^{B\to RA}(\psi^B)\right]$. Denote $\phi^{RA}\eqdef\mV^{B\to RA}(\psi^B)$. Then,
\ba
M_1(\psi^B)&\geq M_1(\phi^{RA})\\
&\geq \inf\Big\{M_1(\chi^{RA})\;:\;\tr_{R}\left[\chi^{RA}\right]=\rho^A\;\;,\;\;\chi\in\pure(RA)\Big\}
\ea
Since the inequality above holds for all $\psi^B$ such that there exists $\mE\in\mf(B\to A)$ satisfying $\rho^A=\mE(\psi^B)$, we must have
\be
\overline{M}_1(\rho^A)\geq\inf\Big\{M_1(\psi^{RA})\;:\;\tr_{R}\left[\psi^{RA}\right]=\rho^A\;\;,\;\;\psi\in\pure(RA)\Big\}\;.
\ee
This completes the proof.
\end{proof}

\subsubsection{Example: The Purified Distance}

An important distance measure that is use quite often in single-shot quantum information theory is the purified distance.
In this example we will see that the purified distance is the maximal extension of the trace distance from pure states to mixed states. This means that any other distance measures that satisfies the DPI and that reduces to the trace distance on pure states must be no greater than the purified distance. The purified distance on mixed normalized states is defined by
\be
P(\rho,\sigma)\eqdef\inf_{\psi,\phi}D(\psi^{RA},\phi^{RA})=\sqrt{1-F(\rho,\sigma)^2}\quad\quad\forall\;\rho,\sigma\in\md(A)\;,
\ee
where the infimum is over all purifications $\psi^{RA}$ and $\phi^{RA}$ of $\rho^{A}$ and $\sigma^A$, respectively, $D$ is the trace distance, and $F$ is the fidelity. 

\begin{corollary*}
The maximal extension of the trace distance from normalized pure states to sub-normalized mixed states is given by the purified distance
\be
P\left(\trho,\tsigma\right)\eqdef\sqrt{1-\bar{F}(\trho,\tsigma)}\quad\quad\forall\;\trho,\tsigma\in\tilde{\md}(A)\;.
\ee
where $\bar{F}$ is the generalized fidelity. 
\end{corollary*}

\begin{remark}
The significance of this proposition is that any distance-divergence measure on subnormalized states that reduces to the trace distance on pure normalized states must be no greater than the purified distance. Therefore, the purified distance is optimal in this sense.
\end{remark}

\begin{proof}
As discussed in the main text, in the QRT associated with divergences, the resources are pair of states $(\rho,\sigma)$, and the free maps are pairs of two identical channels $(\mE,\mE)$, where $\mE\in\cptp(A\to B)$. Due to Stinespring dilation, for any channel $\mE$ there exists an isometry $\mV^{A\to BE}$ satisfying~\eqref{isoe}. Since the pair $(\mV^{A\to BE},\mV^{A\to BE})$ is also free we conclude that this QRT is purifiable. Hence, the theorem above implies that the maximal extension, $\bD$, of the trace distance, $D$, is given by
\be
\bD(\rho,\sigma)=\inf_{\psi,\phi}D(\psi^{RA},\phi^{RA})=P\left(\rho,\sigma\right)\quad\quad\forall\;\rho,\sigma\in\md(A)
\ee
where the infimum is over all purifications $\psi^{RA}$ and $\phi^{RA}$ of $\rho^{A}$ and $\sigma^A$, respectively.
Combining this with Theorem~\ref{gmain}, completes the proof of the corollary for subnormalized states.
\end{proof}

\subsection{Extensions of Entanglement Measures}

Here we show how the extensions techniques  can be applied directly to entanglement theory. Particularly, we consider the extensions in~\eqref{909} and~\eqref{908} to entanglement theory.
For any function on pure states $E_1:\bigcup_{A,B}\pure(AB)\to\mbb{R}$  we define its maximal  extension to mixed bipartite states as
\be\label{eee1}
\overline{E}_1(\rho^{AB})\eqdef\inf\Big\{E_1(\psi^{A'B'})\;:\;\rho^{AB}=\mE(\psi^{A'B'})\;,\;\mE\in\locc(A'B'\to AB)\;,\;\psi\in\pure(A'B')\Big\}
\ee
and its minimal extension as
\be\label{eee2}
\underline{E}_1(\rho^{AB})\eqdef\sup\Big\{E_1(\psi^{A'B'})\;:\;\psi^{A'B'}=\mE(\rho^{AB})\;,\;\mE\in\locc(AB\to A'B')\;,\;\psi\in\pure(A'B')\Big\}
\ee
In the above definition we took $\mf(AB\to A'B') =\locc(AB\to A'B')$. Theorem~\ref{properties} take the following form in entanglement theory.

\begin{corollary*}
Let $E$ be a measure of entanglement on pure states. 
\begin{enumerate}
\item For any $\psi\in{\rm \pure}(AB)$ 
\be
\overline{E}_1(\psi^{AB})=\underline{E}_1(\psi^{AB})=E_1(\psi^{AB})\;.
\ee
\item For any $\rho\in\md(AB)$ and $\mE\in\locc(AB\to A'B')$
\ba
&\overline{E}_1\left(\mE^{AB\to A'B'}(\rho^{AB})\right)\leq \overline{E}_1\left(\rho^{AB}\right)\\
&\underline{E}_1\left(\mE^{AB\to A'B'}(\rho^{AB})\right)\leq \underline{E}_1\left(\rho^{AB}\right)\;.
\ea
\item For any measure of entanglement $E$ that reduces to $E_1$ on pure states
\be\label{bbb1}
\underline{E}_1\left(\rho^{AB}\right)\leq E\left(\rho^{AB}\right)\leq \overline{E}_1\left(\rho^{AB}\right)\;.
\ee
\end{enumerate}
\end{corollary*}

For any pure bipartite state $\psi\in\pure(AB)$ and any matrix $M$ on system $B$, there exists unitaries $U$ and $V$ such that
\be
I^A\otimes M|\psi^{AB}\ra=UM^T\otimes V|\psi^{AB}\ra\;.
\ee
Therefore, any measurement performed by Bob on a pure bipartite state can be simulated by a measurement performed by Alice, followed by a unitary performed by Bob. This means that the channel $\mE$ in the definition of $\overline{E}_1$ can be expressed as a 1-way LOCC of the form
\be
\mE\left(\sigma^{AB}\right)=\sum_{j}\left(K_j\otimes U_j\right)\sigma^{AB}\left(K_j\otimes U_j\right)^\dag\quad\quad\forall\;\sigma\in\md(AB)\;.
\ee 

\subsubsection{Example: The Schmidt number of mixed bipartite states} 

One common entanglement monotone of pure states is the Schmidt number. It is defined on a pure state $\psi^{AB}$ as
\be
N(\psi^{AB})\eqdef \rank(\psi^A)
\ee
where $\psi^A$ is the reduced density matrix of $\psi^{AB}$. Note that the definition above remains unchanged on sub-normalized pure bipartite states. In~\cite{Terhal-2000a} the Schmidt number of density matrices was define as follows.

\begin{definition*}[\cite{Terhal-2000a}]\label{SN}
 A bipartite density matrix $\rho\in\md(AB)$ has Schmidt
number $k$ if (1) for any decomposition of $\rho^{AB}=\sum_ip_i|\psi_i\lr\psi_i|$, $p_i\geq 0$
at least one of the vectors $\psi_i$ has at least
Schmidt rank $k$ and (2) there exists a decomposition of $\rho^{AB}$
with all vectors $\psi_i^{AB}$ of Schmidt rank at most $k$.
\end{definition*}

Let $\rho\in\md(AB)$ be a bipartite mixed state. Then, the maximal and minimal extensions of $N$ to mixed state are given by
\begin{align}
&\overline{N}(\rho^{AB})\eqdef \inf\Big\{N\left(\psi^{A'B'}\right)\;:\;\rho^{AB}=\mE^{A'B'\to AB}(\psi^{A'B'})\;\;,\;\;\mE\in\locc(A'B'\to AB)\;\;,\;\;\psi\in\pure(A'B')\Big\}\\
&\underline{N}(\rho^{AB})\eqdef\sup\Big\{N\left(\mE^{AB\to A'B'}(\rho^{AB})\right)\;:\;\mE\in\locc(AB\to A'B')\;\;,\;\;\mE^{AB\to A'B'}(\rho^{AB})\in\pure(A'B')\Big\}
\end{align}

\begin{lemma}
The maximal extension $\overline{N}$ equals the Schmidt number, $N$, as defined above.
\end{lemma}
\begin{proof}
Since $N$ is an entanglement measure on mixed bipartite states, Theorem~\ref{properties} implies that for any $\rho\in\md(AB)$
\be
{N}(\rho^{AB})\leq \overline{N}(\rho^{AB})\;.
\ee
For the other direction, observe first that due to teleportation, $\overline{N}(\rho^{AB})\leq k\eqdef \min\{|A|,|B|\}$ since by taking $\psi^{A'B'}$ to be the maximally entangled state in $\md(AB)$ with $A'=A$ and $B'=B$ we get that there exists a map $\mE\in\locc(A'B'\to AB)$ that take $\psi^{AB}$ to $\rho^{AB}$. Now, denote by $k'\eqdef N(\psi^{A'B'})$ and by $\phi_{+}^{k'}$ the maximally entangled states with Schmidt rank $k'$. Then, w.l.o.g. we can assume that $\psi^{A'B'}=\phi_{+}^{k'}$ (recall that $N(\psi^{A'B'})=N(\phi^{k'}_+)$). We summarize all of this as
\be
\overline{N}(\rho^{AB})=\min\Big\{k\in\mbb{N}\;:\;\rho^{AB}=\mE^{A'B'\to AB}(\phi^{k}_{+})\;\;,\;\;\mE\in\locc(A'B'\to AB)\Big\}
\ee
Let $k$ be optimal (i.e. $k=\overline{N}(\rho^{AB})$), and note that the condition $$\rho^{AB}=\mE^{A'B'\to AB}(\phi^{k}_{+})=\sum_j(K_j\otimes M_j)|\phi_+^k\lr\phi_+^k|(K_j\otimes M_j)^\dag$$ implies that there exists a pure state decomposition of $\rho=\sum_jp_j\psi_j$ with the property that 
$N(\psi_j^{AB})\leq k$ for all $j$. Moreover, suppose by contradiction, that there exists a decomposition of $\rho^{AB}=\sum_jq_j\phi_j^{AB}$ with the property that all the states $\phi_j$ have Schmidt rank $m<k$. Then, there exists matrices $K_j$ such that $|\phi_j\ra=K_j\otimes I^B|\phi^m_+\ra$, and therefore, in particular, there exists  $\mE\in\locc(A'B'\to AB)$
such that $\rho=\mE(\phi_{+}^m)$ in contradiction with the optimality of $k$. Hence, all pure state decompositions of $\rho^{AB}$ must contain at least one state with Schmidt rank at least $k$. That is, $\overline{N}(\rho^{AB})={N}(\rho^{AB})$.
\end{proof}

\begin{corollary*}
Let $E$ be an entanglement measure  that reduces to the Schmidt number on pure bipartite states.
Then, for all $\rho\in\md(AB)$
\be
E(\rho^{AB})\leq N(\rho^{AB})\;.
\ee
\end{corollary*}

\subsubsection{Smoothed Extensions}

The maximal and minimal extensions of a pure-state entanglement measure $E$, can be interpreted as the zero-error pure-entanglement cost and distillation of $\rho^{AB}$, respectively. Note also that the minimal extension is quite often zero since it is not alway possible to find an LOCC channel that can be used to obtain a pure entangled state from $\rho^{AB}$. To avoid that, one can smooth these functions to get
\ba
\overline{E}_\epsilon(\rho^{AB})&\eqdef\min_{\rho'\in\mb_\epsilon(\rho)}\underline{E}(\rho'^{AB})\\
&=
\inf\Big\{E(\psi^{A'B'})\;:\;P\left(\rho^{AB},\mE(\psi^{A'B'})\right)\leq\epsilon\;,\;\mE\in\locc(A'B'\to AB)\;,\;\psi\in\pure(A'B')\Big\}
\ea
where $P$ is the purified distance, and similarly
\be
\underline{E}_\epsilon(\rho^{AB})=\sup\Big\{E(\psi^{A'B'})\;:\;P\left(\psi^{A'B'},\mE(\rho^{AB})\right)\leq\epsilon\;,\;\mE\in\locc(AB\to A'B')\;,\;\psi\in\pure(A'B')\Big\}\;.
\ee
These smoothed quantities can be interpreted as the $\epsilon$-error one-shot pure-entanglement cost and distillation, respectively. Due to the following lemma, these quantities are themselves measures of entanglement.

\begin{lemma*}
Let $(\mr,\mf)$ be a GRT, $\epsilon>0$, $\mr_1(A)\subseteq\mr(A)$ for any system $A$, and $M:\bigcup_A\mr_1(A)\to\mbb{R}$ be  an $\mr_1$-resource measure.
Then, for any $\rho\in\mr(A)$ and any $\mE\in\mf(A\to B)$
\be
\uM_\epsilon\left(\mE^{A\to B}(\rho^A)\right)\leq \uM_\epsilon\left(\rho^A\right)\quad\text{and}\quad
\bM_\epsilon\left(\mE^{A\to B}(\rho^A)\right)\leq \bM_\epsilon\left(\rho^A\right)\;.
\ee
\end{lemma*}
\begin{proof}
For $\mE\in\mf(A\to B)$
\ba
\bM_\epsilon\left(\mE^{A\to B}(\rho^A)\right)&=\min_{\substack{P(\mE(\rho),\sigma)\leq\epsilon\\\sigma\in\md(B)}}\bM(\sigma)\\
&\leq \min_{\substack{P(\mE(\rho),\mE(\omega))\leq\epsilon\\\omega\in\md(A)}}\bM(\mE(\omega))\\
&\leq \min_{\substack{P(\mE(\rho),\mE(\omega))\leq\epsilon\\\omega\in\md(A)}}\bM(\omega)\\
&\leq \min_{\substack{P(\rho,\omega)\leq\epsilon\\\omega\in\md(A)}}\bM(\omega)\\
&=\bM_\epsilon\left(\rho^A\right)\;.
\ea
where the first inequality follows by restriction $\sigma$ to be of the form $\mE(\omega)$, the second inequality from monotonicity of $\bM$, and the last inequality from the DPI of the purified distance.
The proof of the monotonicity of $\uM_\epsilon$ follows similar lines.
\end{proof}

\subsubsection{Extensions of entanglement monotones}

Entanglement monotones on pure states are functions that do not increase on average under LOCC. Here we characterize ensemble of pure bipartite states, $\{p_x,\psi_x^{AB}\}$, in terms of a classical quantum state
\be\label{purecq}
\rho^{XAB}\eqdef\sum_xp_x|x\lr x|^X\otimes\psi_x^{AB}\;,
\ee
where the classical flag system $X$ is held by Alice or Bob. With these notations $E:\bigcup_{A,B}\pure(AB)\to\mbb{R}$ is called a pure state entanglement monotone if 
\be
E\left(\psi^{AB}\right)\geq \sum_xp_xE\left(\psi^{AB}_x\right)\;,
\ee
whenever $\psi^{AB}$ can be converted to $\rho^{XAB}$ by LOCC. We can therefore extend the definition of $E$ to classical quantum states of the form~\eqref{purecq} via
\be
E\left(\rho^{XAB}\right)\eqdef\sum_{x}p_xE\left(\psi^{AB}_x\right)\;.
\ee
In order to extend $E$ to arbitrary bipartite states in $\md(AB)$ we define $\mr_1(XAB)$ the set of all cq-states of the form~\eqref{purecq}. Then, the optimal extensions of $E$ to all bipartite states are defined by
\begin{align}
&\overline{E}(\rho^{AB})\eqdef\inf\Big\{E(\sigma^{YA'B'})\;:\;\rho^{AB}=\mE(\sigma^{YA'B'})\;,\;\mE\in\locc(YA'B'\to AB)\;,\;\sigma\in\mf_1(YA'B')\Big\}\label{ee1}\\
&\underline{E}(\rho^{AB})\eqdef\sup\Big\{E(\sigma^{YA'B'})\;:\;\sigma^{YA'B'}=\mE(\rho^{AB})\;,\;\mE\in\locc(AB\to YA'B')\;,\;\sigma\in\mf_1(YA'B')\Big\}\label{ee2}
\end{align}
From Theorem~\ref{properties} it follows that both $\overline{E}$ and $\underline{E}$ reduces to $E$ on cq-states of the form~\ref{purecq}, they are non-increasing under LOCC, and any for any other entanglement measure $M$ that reduces to $E$ on cq-states of the form~\eqref{purecq} satisfies
\be\label{bbb2}
\underline{E}\left(\rho^{AB}\right)\leq M\left(\rho^{AB}\right)\leq \overline{E}\left(\rho^{AB}\right)\;.
\ee
The significance of the bounds above as compared with~\eqref{bbb2} is that $\overline{E}$ and $\underline{E}$ as defined in~\eqref{ee1} and~\eqref{ee2}
provide a tighter bound than the ones defined in~\eqref{eee1} and~\eqref{eee2}.

\end{document}